\def\year{2022}\relax
\documentclass[letterpaper]{article} 
\usepackage{aaai22}  
\usepackage{times}  
\usepackage{helvet}  
\usepackage{courier}  
\usepackage[hyphens]{url}  
\usepackage{graphicx} 
\usepackage{natbib}  
\usepackage{caption} 
\DeclareCaptionStyle{ruled}{labelfont=normalfont,labelsep=colon,strut=off} 
\usepackage{newfloat}
\usepackage{listings}

\usepackage[utf8]{inputenc} 
\usepackage[T1]{fontenc}    
\usepackage{url}            
\usepackage{booktabs}       
\usepackage{amsfonts}       
\usepackage{nicefrac}       
\usepackage{microtype}      
\usepackage{bm}
\usepackage{caption}
\usepackage{algorithm}
\usepackage[noend]{algpseudocode}
\usepackage{amssymb}
\usepackage{mathtools}
\usepackage{amsmath}
\usepackage{tikz}
\usetikzlibrary{calc}
\usepackage{todonotes}
\usepackage{dsfont}
\usepackage{enumitem}
\usepackage{tikz}
\usepackage{amsthm}
\usepackage{thm-restate}
\usepackage{multirow}
\usepackage{newfloat}
\usepackage{wrapfig}
\usepackage{dsfont}
\usepackage{comment}
\usepackage{multirow}
\usepackage{graphicx}
\usepackage{mleftright}
\usepackage{bbm}
\usepackage{thm-restate}
\usepackage[switch]{lineno}
\makeatletter
\AtBeginDocument{%
	\@ifpackageloaded{amsmath}{%
		\newcommand*\patchAmsMathEnvironmentForLineno[1]{%
			\expandafter\let\csname old#1\expandafter\endcsname\csname #1\endcsname
			\expandafter\let\csname oldend#1\expandafter\endcsname\csname end#1\endcsname
			\renewenvironment{#1}%
			{\linenomath\csname old#1\endcsname}%
			{\csname oldend#1\endcsname\endlinenomath}%
		}%
		\newcommand*\patchBothAmsMathEnvironmentsForLineno[1]{%
			\patchAmsMathEnvironmentForLineno{#1}%
			\patchAmsMathEnvironmentForLineno{#1*}%
		}%
		\patchBothAmsMathEnvironmentsForLineno{equation}%
		\patchBothAmsMathEnvironmentsForLineno{align}%
		\patchBothAmsMathEnvironmentsForLineno{flalign}%
		\patchBothAmsMathEnvironmentsForLineno{alignat}%
		\patchBothAmsMathEnvironmentsForLineno{gather}%
		\patchBothAmsMathEnvironmentsForLineno{multline}%
	}{}
}
\makeatother

\newtheorem{example}{Example}

\newtheorem{lemma}{Lemma}
\newtheorem{remark}{Remark}
\newtheorem{corollary}{Corollary}

\newtheorem{definition}{Definition}

\newcommand{\sset}{\mathcal{S}}

\newcommand{\pr}{\textnormal{Pr}}
\newcommand*\circled[1]{\tikz[baseline=(char.base)]{
		\node[shape=circle,draw,inner sep=1pt] (char) {\scriptsize #1};}}

\definecolor{pi}{RGB}{255,221,244}

\DeclareMathOperator*{\argmax}{arg\,max}

\newcommand{\N}{\mathcal{N}}
\newcommand{\V}{\mathcal{V}}
\newcommand{\Rev}{\textsc{Rev}}

\title{Signaling in Posted Price Auctions}

\floatstyle{ruled}
\newfloat{listing}{tb}{lst}{}
\floatname{listing}{Listing}
\title{Signaling in Posted Price Auctions}

\author{
	Matteo Castiglioni, Giulia Romano, Alberto Marchesi, Nicola Gatti \\
}	
\affiliations{
	Politecnico di Milano, Piazza Leonardo da Vinci 32, I-20133, Milan, Italy \\
	\{name.surname\}@polimi.it
}

\begin{document}


\maketitle

\begin{abstract}
	 We study single-item single-unit \emph{Bayesian posted price auctions}, where buyers arrive sequentially and their valuations for the item being sold depend on a random, unknown state of nature.
	 The seller has complete knowledge of the actual state and can send signals to the buyers so as to disclose information about it.
	 For instance, the state of nature may reflect the condition and/or some particular features of the item, which are known to the seller only.
	 The problem faced by the seller is about \emph{how to partially disclose information about the state so as to maximize revenue}.
	 Unlike classical signaling problems, in this setting, the seller must also correlate the signals being sent to the buyers with some price proposals for them.
	 %
	 This introduces additional challenges compared to standard settings.
	 %
	 %
	 %
	 We consider two cases: the one where the seller can only send signals \emph{publicly} visible to all buyers, and the case in which the seller can \emph{privately} send a different signal to each buyer.
	 As a first step, we prove that, in both settings, the problem of maximizing the seller's revenue does \emph{not} admit an FPTAS unless $\mathsf{P}=\mathsf{NP}$, even for basic instances with a single buyer.
	 As a result, in the rest of the paper, we focus on designing PTASs.
	 In order to do so, we first introduce a unifying framework encompassing both public and private signaling, whose core result is a decomposition lemma that allows focusing on a finite set of possible buyers' posteriors.
	 This forms the basis on which our PTASs are developed.
	 In particular, in the public signaling setting, our PTAS employs some \emph{ad hoc} techniques based on linear programming,
	 %
	 while our PTAS for the private setting relies on the ellipsoid method to solve an exponentially-sized LP in polynomial time.
	 %
	 %
	 %
	 In the latter case, we need a custom approximate separation oracle, which we implement with a dynamic programming approach.
	 %
	 %
\end{abstract}

\section{Introduction}


In \emph{posted price auctions}, the seller tries to sell an item by proposing \emph{take-it-or-leave-it} prices to buyers arriving sequentially.
Each buyer has to choose between declining the offer---without having the possibility of coming back---or accepting it, thus ending the auction.
Nowadays, posted pricing is the most used selling format in e-commerce~\citep{einav2018auctions}, whose sales reach over \$4 trillion in 2020~\citep{eMarketer}.
Posted price auctions are ubiquitous in settings such as, for example, online travel agencies (\emph{e.g.}, \emph{Expedia}), accommodation websites (\emph{e.g.}, \emph{Booking.com}), and retail platforms (\emph{e.g.}, \emph{Amazon} and \emph{eBay}).
As a result, growing attention has been devoted to their analysis, both in economics~\citep{seifert2006posted} and in computer science~\citep{chawla2010multi,babaioff2015dynamic,babaioff2017posting,adamczyk2017sequential,correa2017posted}, within AI and machine learning in particular~\citep{kleinberg2003value,Shah2019Semi,romano2021online}.

We study \emph{Bayesian} posted price auctions, where the buyers' valuations for the item depend on a random state of nature, which is known to the seller only.
By applying the \emph{Bayesian persuasion} framework~\citep{kamenica2011bayesian}, we consider the case in which the seller (sender) can send signals to the buyers (receivers) so as to disclose information about the state. 
Thus, in a Bayesian auction, the seller does \emph{not} only have to decide price proposals for the buyers, but also \emph{how to partially disclose information about the state so as to maximize revenue}.
%
%
Our model finds application in several real-world scenarios.
For instance, in an e-commerce platform, the state of nature may reflect the condition (or quality) of the item being sold and/or some of its features.
These are known to the seller only since the buyers cannot see the item given that the auction is carried out on the web.

\paragraph{Original Contributions.}
We study the problem of maximizing seller's revenue in single-item single-unit Bayesian posted price auctions, focusing on two different settings: \emph{public signaling}, where the signals are publicly visible to all buyers, and \emph{private signaling}, in which the seller can send a different signal to each buyer through private communication channels.
As a first negative result, we prove that, in both settings, the problem does \emph{not} admit an FPTAS unless $\mathsf{P}=\mathsf{NP}$, even for basic instances with a single buyer.
Then, we provide tight positive results by designing a PTAS for each setting.
In order to do so, we first introduce a unifying framework encompassing both public and private signaling.
Its core result is a \emph{decomposition lemma} that allows us to focus on a finite set of buyers' posterior beliefs over states of nature---called $q$-uniform posteriors---, rather than reasoning about signaling schemes with a (potentially) infinite number of signals.
Compared to previous works on signaling, our framework has to deal with some additional challenges.
The main one is that, in our model, the seller (sender) is \emph{not} only required to choose how to send signals, but they also have to take some actions in the form of price proposals.
This requires significant extensions to standard approaches based on decomposition lemmas~\citep{ChengMixture2015,XuTractability2020,CastiglioniPersuading2021}.
%
%
The framework forms the basis on which we design our PTASs.
In the public setting, it establishes a connection between signaling schemes and probability distributions over $q$-uniform posteriors.
This allows us to formulate the seller's revenue-maximizing problem as an LP of polynomial size, whose objective coefficients are \emph{not} readily available. However, they can be approximately computed in polynomial time by an algorithm for finding approximately-optimal prices in (non-Bayesian) posted price auctions, which may also be of independent interest.
Solving the LP with approximate coefficients then gives the desired PTAS.
As for the private setting, our framework provides a connection between marginal signaling schemes of each buyer and probability distributions over $q$-uniform posteriors, which, to the best of our knowledge, is the first of its kind, since previous works are limited to public settings~\citep{ChengMixture2015,castiglioni2020Public}.\footnote{A notable exception is~\citep{CastiglioniPersuading2021}, which studies a specific case in between private and public signaling schemes.}
Such connection allows us to formulate an LP correlating marginal signaling schemes together and with price proposals.
Although the LP has an exponential number of variables, we show that it can still be approximately solved in polynomial time by means of the ellipsoid method.
This requires the implementation of a problem-specific \emph{approximate separation oracle} that can be implemented in polynomial time by means of a dynamic programming algorithm.

\paragraph{Related Works.}
The computational study of Bayesian persuasion has received terrific attention \citep{vasserman2015implementing,castiglioni2019persuading,rabinovich2015information,candogan2019persuasion,castiglioni2022bayesian, castiglioni2020signaling}.
The works most related to ours are those addressing second-price auctions.
\citet{emek2014signaling} provide an LP to compute an optimal public signaling scheme in the known-valuation setting, and they show that the problem is $\mathsf{NP}$-hard in the Bayesian setting.
\citet{ChengMixture2015} provide a PTAS for this latter case.
\citet{bacchio2022} extend the framework to study ad auctions with Vickrey–Clarke–Groves payments.
Finally, \citet{badanidiyuru2018targeting} focus on the design of algorithms whose running time is independent from the number of states of nature.	
They initiate the study of private signaling, showing that, 
in second-price auctions, it may introduce non-trivial equilibrium selection issues.

\section{Preliminaries}\label{sec:prelim}

%

\subsection{Bayesian Posted Price Auctions and Signaling}\label{subsec:signaling_auction}

In a \emph{posted price auction}, the seller tries to sell an item to a finite set $\N \coloneqq \{1,\dots,n\}$ of buyers arriving sequentially according to a fixed ordering.
W.l.o.g., we let buyer $i \in \N$ be the $i$-th buyer according to such ordering.
%
%
%
%
The seller chooses a price proposal $p_i \in [0,1]$ for each buyer $i \in \N$.
%
%
%
%
Then, each buyer in turn has to decide whether to buy the item for the proposed price or not.
Buyer $i \in \N$ buys only if their item valuation is at least the proposed price $p_i$.\footnote{As customary in the literature, we assume that buyers always buy when they are offered a price that is equal to their valuation.}
In that case, the auction ends and the seller gets revenue $p_i$ for selling the item, otherwise the auction continues with the next buyer.
%
%
%

We study \emph{Bayesian} posted price auctions, characterized by a finite set of $d$ states of nature, namely $\Theta \coloneqq \{ \theta_1, \ldots,\theta_d \}$.
%
Each buyer $i \in \N$ has a valuation vector $v_i \in [0,1]^d$, with $v_i(\theta)$ representing buyer $i$'s valuation when the state is $\theta \in \Theta$.
Each valuation $v_i$ is independently drawn from a probability distribution $\V_i$ supported on $[0,1]^d$.
For the ease of presentation, we let ${V} \in [0,1]^{n \times d}$ be the matrix of buyers' valuations, whose entries are ${V}(i,\theta) \coloneqq {v}_i(\theta)$ for all $i \in \mathcal{N}$ and $\theta \in \Theta$.\footnote{Sometimes, we also write $V_i \coloneqq v_i^\top$ to denote the $i$-th row of matrix $V$, which is the valuation of buyer $i \in \N$.}
Moreover, by letting $\V \coloneqq \{ \V_i \}_{i \in \N}$ be the collection of all distributions of buyers' valuations, we write $V \sim \V$ to denote that $V$ is built by drawing each $v_i$ independently from $\V_i$. 
%
%

We model signaling with the \emph{Bayesian persuasion} framework by~\citet{kamenica2011bayesian}.
We consider the case in which the seller---having knowledge of the state of nature---acts as a \emph{sender} by issuing signals to the buyers (the \emph{receivers}), so as to partially disclose information about the state and increase revenue. 
As customary in the literature, we assume that the state is drawn from a common prior distribution $\mu \in \Delta_{\Theta}$, explicitly known to both the seller and the buyers.\footnote{In this work, given a finite set $X$, we denote with $\Delta_X$ the ($|X|-1$)-dimensional simplex defined over the elements of $X$.}
We denote by $\mu_\theta$ the probability of state $\theta \in \Theta$.
The seller \emph{commits to a signaling scheme} $\phi$, which is a randomized mapping from states of nature to signals for the receivers.
Letting $\sset_i$ be the set of signals for buyer $i \in \N$, a signaling scheme is a function $\phi:\Theta\to\Delta_{\sset}$, where $\sset \coloneqq \bigtimes_{i \in \N} \sset_i$.
An element $s \in \sset$---called \emph{signal profile}---is a tuple specifying a signal for each buyer.
We use $s_i$ to refer to the $i$-th component of any $s \in \sset$ (\emph{i.e.}, the signal for buyer $i$), so that $s = (s_1, \ldots, s_n)$.
%
%
We let $\phi_\theta(s)$ be the probability of drawing signal profile $s\in \sset$ when the state is $\theta \in \Theta$.
%
%
Furthermore, we let $\phi_i : \Theta \to \Delta_{\sset_i}$ be the marginal signaling scheme of buyer $i \in \N$, with $\phi_i(\theta)$ being the marginalization of $\phi(\theta)$ with respect to buyer $i$'s signals.
%
As for general signaling schemes,
$\phi_{i,\theta}(s_i) \coloneqq \sum_{s' \in \sset: s'_i=s_i}\phi_{\theta}(s')$ denotes the probability of drawing signal $s_i \in \sset_i$ when the state is $\theta \in \Theta$.

%
Price proposals may depend on the signals being sent to the buyers.
Formally, the seller \emph{commits to a price function} $f: \sset \to [0,1]^n$, with $f(s) \in [0,1]^n$ being the price vector when the signal profile is $s \in \sset$.
%
%
We assume that prices proposed to buyer $i$ only depend on the signals sent to them, and \emph{not} on the signals sent to other buyers.
Thus, w.l.o.g., we can work with functions $f_i : \sset_i \to [0,1]$ defining prices for each buyer $i\in \N$ independently, with $f_i(s_i)$ denoting the $i$-th component of $f(s)$ for all $s \in \sset$ and $i \in \N$.\footnote{Let us remark that our assumption on the seller's price function ensures that a buyer does \emph{not} get additional information about the state of nature by observing the proposed price, since the latter only depends on the signal which is revealed to them anyway.}
%
%
%
%

\begin{figure}[!htp]
	\centering
	\includegraphics[width=0.47\textwidth]{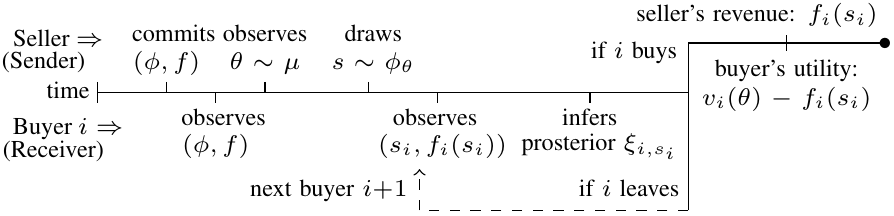}
	\caption{Interaction between the seller and the buyers.}
	\label{fig:inter}
\end{figure}

%
The interaction involving the seller and the buyers goes on as follows (Figure~\ref{fig:inter}): (i) the seller commits to a signaling scheme $\phi: \Theta \to\Delta_{\sset}$ and a price function $f: \sset \to [0,1]^n$, and the buyers observe such commitments; (ii) the seller observes the state of nature $\theta \sim \mu$; (iii) the seller draws a signal profile $s \sim \phi(\theta)$; and (iv) the buyers arrive sequentially, with each buyer $i \in \N$ observing their signal $s_i$ and being proposed price $f_i(s_i)$.
Then, each buyer rationally updates their prior belief over states according to Bayes rule, and buys the item only if their expected valuation for the item is greater than or equal to the offered price.
The interaction terminates whenever a buyer decides to buy the item or there are no more buyers arriving.
The following paragraph formally defines the elements involved in step (iv).

\paragraph{Buyers' Posteriors.}
In step (iv),
%
%
%
%
a buyer $i \in \N$ receiving a signal $s_i \in \sset_i$ infers a posterior belief over states (also called \emph{posterior}), which we denote by $\xi_{i,s_i} \in \Delta_{\Theta}$, with $\xi_{i,s_i}(\theta)$ being the posterior probability of state $\theta \in \Theta$.
%
%
Formally,
\begin{equation}\label{eq:posterior}
	\xi_{i,s_i}(\theta) \coloneqq \frac{{\mu}_\theta {\phi}_{i, \theta}(s_i) }{\sum_{\theta'\in\Theta}{\mu}_{\theta'} {\phi}_{i, \theta'}(s_i) } .
\end{equation}
Thus, after receiving signal $s_i \in \sset_i$, buyer $i$'s expected valuation for the item is $\sum_{\theta \in \Theta} v_i(\theta) \, \xi_{i,s_i}(\theta)$, and the buyer buys it only if such value is at least as large as the price $f_i(s_i)$. 
In the following, given a signal profile $s \in \sset$, we denote by $\xi_s$ a tuple defining all buyers' posteriors resulting from observing signals in $s$; formally, $\xi_s \coloneqq ({\xi}_{1, s_1}, \ldots, {\xi}_{n, s_n})$.
%
%
%
%
%
%

\paragraph{Distributions on Posteriors.}
In single-receiver Bayesian persuasion models, it is oftentimes useful to represent signaling schemes as convex combinations of the posteriors they can induce.
In our setting, a marginal signaling scheme $\phi_i: \Theta \to \Delta_{\sset_i}$ of buyer $i \in \N$ \emph{induces} a probability distribution $\gamma_i$ over posteriors in $\Delta_{\Theta}$, with $\gamma_i(\xi_i)$ denoting the probability of posterior $\xi_i \in \Delta_{\Theta}$.
Formally, it holds that 
\[
{\gamma}_i({\xi}_i)  \coloneqq \sum_{s_i \in \mathcal{S}_i:  {\xi}_{i,s_i} = {\xi}_{i} } \sum_{\theta \in \Theta} {\mu}_\theta {\phi}_{i, \theta}(s_i) .
\]
Intuitively, $\gamma_i({\xi}_i)$ denotes the probability that buyer $i$ has posterior $\xi_i$. 
Indeed, it is possible to directly reason about distributions $\gamma_i$ rather than marginal signaling schemes, provided that such distributions are \emph{consistent} with the prior.
Formally, by letting $\text{supp}(\gamma_i) \coloneqq \{\xi_i \in \Delta_{\Theta} \mid \gamma_i(\xi_i) > 0 \}$ be the support of $\gamma_i$, it must be required that
\begin{equation}\label{eq:w_consistent}
	\sum_{\xi_i \in \text{supp}(\gamma_i)} \gamma_i(\xi_i) \, \xi_i(\theta)=\mu_\theta \quad \forall \theta \in \Theta.
\end{equation}

\subsection{Computational Problems}\label{subsec:problem_definition}
We focus on the problem of computing a signaling scheme $\phi: \Theta \to \Delta_{\sset}$ and a price function $f: \sset \to [0,1]^n$ that maximize the seller's expected revenue, considering both public and private signaling settings.\footnote{Formally, a signaling scheme $\phi : \Theta \to \Delta_{\sset}$ is \emph{public} if: (i) $\mathcal{S}_i = \mathcal{S}_j$ for all $i, j \in \mathcal{N}$; and (ii) for every $\theta \in \Theta$, ${\phi}_\theta(s) > 0$ only for signal profiles $s \in \mathcal{S}$ such that $s_i = s_j$ for $i, j \in \mathcal{N}$. Since, given a signal profile $s \in \sset$, under a public signaling scheme all the buyers always share the same posterior (\emph{i.e.}, $\xi_{i,s_i} = \xi_{j,s_j}$ for all $i,j \in \N$), we overload notation and sometimes use $\xi_s \in \Delta_{\Theta}$ to denote the unique posterior appearing in $\xi_s = (\xi_{1,s_1,},\ldots,\xi_{n,s_n})$. Similarly, in the public setting, given a posterior $\xi \in \Delta_{\Theta}$ we sometimes write $\xi$ in place of a tuple of $n$ copies of $\xi$.}

We denote by $\textsc{Rev}(\mathcal{V},{p}, \xi)$ the expected revenue of the seller when the distributions of buyers' valuations are given by $\V = \{ \V_i \}_{i \in \N}$, the proposed prices are defined by the vector $p \in  [0,1]^n$, and the buyers' posteriors are those specified by the tuple $\xi = (\xi_1,\ldots,\xi_n)$ containing a posterior $\xi_i \in \Delta_{\Theta}$ for each buyer $i \in \N$.
Then, the seller's expected revenue is:
\[
\sum_{\theta \in \Theta} {\mu}_\theta \sum_{s \in \mathcal{S}} {\phi}_\theta(s) \textsc{Rev} \left( \mathcal{V},f(s), {\xi}_s \right).
\]
In the following, we denote by $OPT$ the value of the seller's expected revenue for a revenue-maximizing $(\phi,f)$ pair. 

In this work, we assume that algorithms have access to a black-box oracle to sample buyers' valuations according to the probability distributions specified by $\V$ (rather than actually knowing such distributions).
Thus, we look for algorithms that output pairs $(\phi,f)$ such that
\[
	\mathbb{E} \left[ \sum_{\theta \in \Theta} {\mu}_\theta \sum_{s \in \mathcal{S}} {\phi}_\theta(s) \textsc{Rev}(\mathcal{V},f(s), {\xi}_s) \right] \geq OPT - \lambda,
\]
where $\lambda \geq 0$ is an additive error.
Notice that the expectation above is with respect to the randomness of the algorithm, which originates from using the black-box sampling oracle.

\section{Hardness of Signaling with a Single Buyer}\label{sec:hardness}

We start with a negative result: there is no FPTAS for the problem of computing a revenue-maximizing $(\phi,f)$ pair unless $\mathsf{P}=\mathsf{NP}$, in both public and private signaling settings.
Our result holds even in the basic case with only one buyer, where public and private signaling are equivalent.
Notice that, in the reduction that we use to prove our result, we assume that the support of the distribution of valuations of the (single) buyer is finite and that such distribution is perfectly known to the seller.
This represents an even simpler setting than that in which the seller has only access to a black-box oracle returning samples drawn from the buyer's distribution of valuations.
%
%
%
The result formally reads as follows:

\begin{restatable}{theorem}{hardness}\label{thm:hardness}
	There is no additive FPTAS for the problem of computing a revenue-maximizing $(\phi,f)$ pair unless $\mathsf{P}=\mathsf{NP}$, even when there is a single buyer.
	%
\end{restatable}

\section{Unifying Public and Private Signaling}\label{sec:general}

In this section, we introduce a general mathematical framework related to buyers' posteriors and distributions over them, proving some results that will be crucial in the rest of this work, both in public and private signaling scenarios.

One of the main difficulties in computing sender-optimal signaling schemes is that they might need a (potentially) infinite number of signals, resulting in infinitely-many receiver's posteriors.
The trick commonly used to circumvent this issue in settings with a finite number of valuations is to use direct signals, which explicitly specify action recommendations for each receiver's valuation~\cite{castiglioni2020online,castiglioni2021multi}.
However, in our auction setting, this solution is \emph{not} viable, since a direct signal for a buyer $i \in \N$ should represent a recommendation for every possible $v_i \in [0,1]^d$, and these are infinitely many.
An alternative technique, which can be employed in our setting, is to restrict the number of possible posteriors.
%

Our core idea is to focus on a small set of posteriors, which are those encoded as particular $q$-uniform probability distributions, as formally stated in the following definition.\footnote{In all the definitions and results of this section (Section~\ref{sec:general}), we denote by $\xi\in \Delta_{\Theta}$ a generic posterior common to all the buyers
	 and with $\gamma$ a probability distribution over $\Delta_{\Theta}$ (\emph{i.e}, over posteriors).}
\begin{definition}[$q$-uniform posterior]\label{def:s_distribution}
	A posterior $\xi \in \Delta_{\Theta}$ is {\em $q$-uniform} if it can be obtained by averaging the elements of a multiset defined by $q\in \mathbb{N}_{>0}$ canonical basis vectors of $\mathbb{R}^d$.
\end{definition}

We denote the set of all $q$-uniform posteriors as $\Xi^q \subset \Delta_{\Theta}$.
%
Notice that the set $\Xi^q$ has size $|\Xi^q| = O \left( d^q \right)$.

The existence of an approximately-optimal signaling scheme that only uses $q$-uniform posteriors is usually proved by means of so-called \emph{decomposition lemmas} (see~\citep{ChengMixture2015,XuTractability2020,CastiglioniPersuading2021}).
The goal of these lemmas is to show that, given some signaling scheme encoded as a distribution over posteriors, it is possible to obtain a new signaling scheme whose corresponding distribution is supported only on $q$-uniform posteriors, and such that the sender's utility only decreases by a small amount.
At the same time, these lemmas must also ensure that the distribution over posteriors corresponding to the new signaling scheme is still consistent (according to Equation~\eqref{eq:w_consistent}).

The main result of our framework (Theorem~\ref{thm:decomposition}) is a decomposition lemma that is suitable for our setting.
Before stating the result, we need to introduce some preliminary definitions.
\begin{definition}[$(\alpha,\epsilon)$-decreasing distribution]\label{def:perturbation}
	Let $\alpha, \epsilon > 0$. A probability distribution $\gamma$ over $\Delta_{\Theta}$ is {\em $(\alpha,\epsilon)$-decreasing around a given posterior $\xi \in \Delta_{\Theta}$} if the following condition holds for every matrix $V \in [0,1]^{n \times d}$ of buyers' valuations:
	\[
		\textnormal{Pr}_{\tilde\xi \sim \gamma} \left\{ V_i  \tilde\xi \ge  V_i \xi -\epsilon \right\} \ge 1-\alpha \quad \forall i \in \N.
	\]
\end{definition}
Intuitively, a probability distribution $\gamma$ as in Definition~\ref{def:perturbation} can be interpreted as a perturbation of the given posterior $\xi$ such that, with high probability, buyers' expected valuations in $\gamma$ are at most $\epsilon$ less than those in posterior $\xi$.\footnote{Definition~\ref{def:perturbation} is similar to analogous ones in the literature~\citep{XuTractability2020,CastiglioniPersuading2021}, where the distance is usually measured in both directions, as $|   V_i \tilde\xi -  V_i \xi | \leq\epsilon$. We look only at the direction of decreasing values, since in a our setting, if a buyer's valuation increases, then the seller's revenue also increases.}

The second definition we need is about functions mapping vectors in $[0,1]^n$---defining a valuation for each buyer---to seller's revenues.
For instance, one such function could be the seller's revenue given price vector $p \in [0,1]^n$.
In particular, we define the stability of a function $g$ compared to another function $h$.
Intuitively, $g$ is stable compared to $h$ if the value of $g$, in expectation over buyers' valuations and posteriors drawn from a probability distribution $\gamma$ that is $(\alpha,\epsilon)$-decreasing around $\xi$, is ``close'' to the the value of $h$ given $\xi$, in expectation over buyers' valuations.\footnote{The notion of compared stability has been already used~\citep{ChengMixture2015,CastiglioniPersuading2021}. However, previous works consider the case in which $g$ is a relaxation of $h$. Instead, our definition is conceptually different, as $g$ and $h$ represent two different functions corresponding to different price vectors of the seller.}
Formally:
\begin{definition}[$(\delta,\alpha,\epsilon)$-stability]\label{def:stability}
	Let $\alpha, \epsilon, \delta  > 0$. Given a posterior $\xi \in \Delta_{\Theta}$, some distributions $\V = \{ \V_i \}_{i \in \N}$, and two functions $g, h: [0,1]^n \to [0,1]$, $g$ is {\em $(\delta,\alpha,\epsilon)$-stable compared to $h$ for $(\xi,\V)$} if, for every probability distribution $\gamma$ over $\Delta_{\Theta}$ that is $(\alpha,\epsilon)$-decreasing  around $\xi$, it holds:
	\[
		\mathbb{E}_{\tilde \xi \sim \gamma,  V \sim \V} \Big[ g( V \tilde\xi) \Big] \geq (1-\alpha)\mathbb{E}_{ V \sim \V}   \Big[ h(  V  \xi) \Big]-\delta\epsilon.
	\]
\end{definition}

Now, we are ready to state our main result.
We show that, for any buyer's posterior $\xi \in \Delta_{\Theta}$, if a function $g$ is stable compare to $h$, then there exists a suitable probability distribution over $q$-uniform posteriors such that the expected value of $g$ given such distribution is ``close'' to that of $h$ given $\xi$.

\begin{restatable}{theorem}{decomposition}\label{thm:decomposition}
	Let $\alpha, \epsilon, \delta > 0$, and set $q \coloneqq \frac{32}{\epsilon^2} \log \frac{4}{\alpha}  $.
	Given a posterior $\xi \in \Delta_{\Theta}$, some distributions $\V = \{ \V_i \}_{i \in \N}$, and two functions $g, h: [0,1]^n \to [0,1]$, if $g$ is $(\delta,\alpha,\epsilon)$-stable compared to $h$ for $(\xi,\V)$, then there exists $\gamma \in \Delta_{\Xi^q}$ such that, for every $\theta \in \Theta$, $\sum_{\tilde\xi \in \textnormal{supp}(\gamma)} \gamma(\tilde\xi) \tilde\xi(\theta) = \xi(\theta)$ and
	%
	\begin{equation}\label{eq:general_bound}
	\mathbb{E}_{\substack{\tilde\xi \sim \gamma \\ V \sim \V}} \Big[ \tilde\xi(\theta) g( V \tilde\xi ) \Big]\hspace{-1mm} \geq \hspace{-1mm} \xi(\theta)\hspace{-1mm} \left[ \hspace{-0.5mm} (1\hspace{-1mm} -\hspace{-0.5mm} \alpha)  \mathbb{E}_{V \sim \V}  \Big[ h(  V \xi ) \Big] \hspace{-1mm} - \hspace{-0.5mm} \delta \epsilon \right] \hspace{-1mm}.\hspace{-0.5mm}
    \end{equation}
	%
\end{restatable}
The crucial feature of Theorem~\ref{thm:decomposition} is that Equation~\eqref{eq:general_bound} holds for every state.
This is fundamental for proving our results in the private signaling scenario.
On the other hand, with public signaling, we will make use of the following (weaker) corollary, obtained by summing Equation~\eqref{eq:general_bound} over all $\theta \in \Theta$.

\begin{corollary}\label{lm:decompositionweak}
	Let $\alpha, \epsilon, \delta > 0$, and set $q \coloneqq \frac{32}{\epsilon^2}  \log \frac{4}{\alpha} $.
	Given a posterior $\xi \in \Delta_{\Theta}$, some distributions $\V = \{ \V_i \}_{i \in \N}$, and two functions $g, h: [0,1]^n \to [0,1]$, if $g$ is $(\delta,\alpha,\epsilon)$-stable compared to $h$ for $(\xi,\V)$, then there exists $\gamma \in \Delta_{\Xi^q}$ such that, for every $\theta \in \Theta$, $\sum_{\tilde\xi \in \textnormal{supp}(\gamma)} \gamma(\tilde\xi) \tilde\xi(\theta) = \xi(\theta)$ and
	%
	\begin{equation}\label{eq:general_boundweak}
	\mathbb{E}_{\tilde\xi \sim \gamma, V \sim \V} \Big[ g( V \tilde\xi ) \Big] \geq (1-\alpha)  \mathbb{E}_{V \sim \V}  \Big[ h(  V \xi) \Big]-\delta \epsilon .
	\end{equation}
	%
	%
\end{corollary}

\section{Warming Up: Non-Bayesian Auctions}\label{sec:non_bayesian}

In this section, we focus on non-Bayesian posted price auctions, proving some results that will be useful in the rest of the paper.\footnote{When we study non-Bayesian posted price auctions, we stick to our notation, with the following differences: valuations are scalars rather than vectors, namely $v_i \in [0,1]$; distributions $\V_i$ are supported on $[0,1]$ rather than $[0,1]^d$; the matrix $V$ is indeed a column vector whose components are buyers' valuations; and the price function $f$ is replaced by a single price vector $p \in [0,1]^n$, with its $i$-th component $p_i$ being the price for buyer $i \in \N$. Moreover, we continue to use the notation $\textsc{Rev}$ to denote seller's revenues, dropping the dependence on the tuple of posteriors. Thus, in a non-Bayesian auction in which the distributions of buyers' valuations are $\V = \{ \V_i \}_{i \in \N}$, the notation $\textsc{Rev}(\V,p)$ simply denotes the seller's expected revenue by selecting a price vector $p \in [0,1]^n$.}
In particular, we study what happens to the seller's expected revenue when buyers' valuations are ``slightly decreased'', proving that the revenue also decreases, but only by a small amount.
This result will be crucial when dealing with public signaling, and it also allows to design a poly-time algorithm for finding approximately-optimal price vectors in non-Bayesian auctions, as we show at the end of this section.

In the following, we extensively use distributions of buyers' valuations as specified in the definition below.
\begin{definition}\label{def:distribs}
	Given $\epsilon > 0$, we denote by $\V= \{ \V_i \}_{i \in \N}$ and $\V^\epsilon= \{ \V^\epsilon_i \}_{i \in \N}$ two collections of distributions of  buyers' valuations such that, for every price vector $p \in [0,1]^n$,
	\[ 
		\pr_{ v_i \sim \V_i^\epsilon} \left\{ v_i \geq p_i-\epsilon \right\} \geq \pr_{v_i \sim \V_i} \left\{   v_i \geq p_i \right\} \quad \forall i \in \N.
	\]
\end{definition}
Intuitively, valuations drawn from $\V^\epsilon$ are ``slightly decreased'' with respect to those drawn from $\V$, since the probability with which any buyer $i \in \N$ buys the item at the (reduced) price $[p_i-\epsilon]_+$ when their valuation is drawn from $\V_i^\epsilon$ is at least as large as the probability of buying at price $p_i$ when their valuation is drawn from $\V_i$.\footnote{In this work, given $x \in \mathbb{R}$, we let $[x]_+ \coloneqq \max \{ x,0 \}$. We extend the $[\cdot]_+$ operator to vectors by applying it component-wise.}

%
Our main contribution in this section (Lemma~\ref{lm:smallDecrease}) is to show that $\max_{p \in [0,1]^n} \textsc{Rev}(\V^\epsilon,p) \geq \max_{p \in [0,1]^n} \textsc{Rev}(\V,p) -\epsilon$.
By letting $p^* \in \argmax_{p \in [0,1]^n} \textsc{Rev}(\V,p)$ be any revenue-maximizing price vector under distributions $\V$, one may na\"ively think that, since under distributions $\V^\epsilon$ and price vector $[p^* - \epsilon]_+$ each buyer would buy the item at least with the same probability as with distributions $\V$ and price vector $p^*$, while paying a price that is only $\epsilon$ less, then $\textsc{Rev}(\V^\epsilon,[p^* - \epsilon]_+) \geq  \textsc{Rev}(\V,p^*)-\epsilon$, proving the result.
However, this line of reasoning does \emph{not} work, as shown by Example~\ref{ex:fail} in the Extended Version.
%
The crucial feature of Example~\ref{ex:fail} is that there exists a $p^*$ in which one buyer is offered a price that is too low, and, thus, the seller prefers \emph{not} to sell the item to them, but rather to a following buyer.
This prevents a direct application of the line of reasoning outlined above, as it shows that incrementing the probability with which a buyer buys is \emph{not} always beneficial.
One could circumvent this issue by considering a $p^*$ such that the seller is never upset if some buyer buys.
In other words, it must be such that each buyer is proposed a price that is at least as large as the seller's expected revenue in the posted price auction restricted to the  following buyers.
%
%
Next, we show that there always exists a $p^*$ with such desirable property.

Letting $\Rev_{> i}(\V,p)$ be the seller's revenue for price vector $p \in [0,1]^n$ and distributions $\V = \{ \V_i  \}_{i \in \N}$ in the auction restricted to buyers $j \in \N : j > i$, we prove the following:
%
%
%
%
\begin{restatable}{lemma}{lemmaone}\label{lm:bigpayments}
	For any $\V = \{ \V_i \}_{i \in \N}$, there exists a revenue-maximizing price vector $p^*\in \argmax_{p \in [0,1]^n} \textsc{Rev}(\V,p)$ such that $p^*_i \geq  \textsc{Rev}_{>i}(\V,p^*)$ for every buyer $i\in \N$.
\end{restatable}

The proof of Lemma~\ref{lm:smallDecrease} builds upon the existence of a revenue-maximizing price vector $p^* \in [0,1]^n$ as in Lemma~\ref{lm:bigpayments} and the fact that, under distributions $\V^\epsilon$, the probability with which each buyer buys the item given price vector $[p^* - \epsilon]_+$ is greater than that with which they would buy given $p^*$.
%
Since the seller's expected revenue is larger when a buyer buys compared to when they do \emph{not} buy (as $p_i^* \ge \textsc{Rev}_{>i}(\V,p^*)$), the seller's expected revenue decreases by at most $\epsilon$.

\begin{restatable}{lemma}{lemmatwo}\label{lm:smallDecrease}
	Given $\epsilon > 0$, let $\V = \{ \V_i \}_{i \in \N}$ and $\V^\epsilon = \{ \V^\epsilon_i \}_{i \in \N}$ satisfying the conditions of Definition~\ref{def:distribs}.
	Then, $\max_{p \in [0,1]^n} \textsc{Rev}(\V^\epsilon,p) \geq \max_{p \in [0,1]^n} \textsc{Rev}(\V,p)-\epsilon$.
	%
	%
\end{restatable}

Lemma~\ref{lm:smallDecrease} will be useful to prove Lemma~\ref{lm:samplingPublic} and to show the compared stability of a suitably-defined function that is used to design a PTAS in the public signaling scenario.

\begin{algorithm}[!htp]
	\caption{\textsc{Find-APX-Prices}}
	\small
	\textbf{Inputs:} \# of samples $K \in \mathbb{N}_{>0}$; \# of discretization steps $b \in \mathbb{N}_{>0}$
	\begin{algorithmic}[1]
		\For{$i \in \N$}
		\For{$k = 1, \ldots, K$}
		\State $v_i^k \gets \text{Sample buyer $i$'s valuation using oracle for $\V_i$}$
		\EndFor
		\State $\V_i^K \gets \text{Empirical distribution of the $K$ i.i.d. samples $v_i^K$}$
		\EndFor
		\State $\V^K \gets \{ \V_i^K \}_{i \in \N}; \quad p \gets \mathbf{0}_n; \quad r \gets 0$
		\For{ $i=n, \ldots, 1$ (in reversed order)}
		\State \hspace{-5.5mm}$p_i \hspace{-1mm}\gets\hspace{-1mm}  \displaystyle\argmax_{p'_i \in P^b} p'_i \pr_{v_i\sim \V_i^K} \hspace{-1mm}\left\{ v_i \geq p'_i \right\} \hspace{-1mm}+\hspace{-1mm} \left( \hspace{-0.5mm} 1 \hspace{-0.5mm} - \hspace{-0.5mm} \pr_{v_i\sim \V_i^K}\hspace{-1mm} \left\{ v_i \geq p'_i \right\} \hspace{-0.5mm} \right) \hspace{-0.5mm} r$
		\State \hspace{-4mm}$r \gets p_i \pr_{v_i\sim \V_i^K} \left\{ v_i \geq p_i \right\} + \left( 1- \pr_{v_i\sim \V_i^K} \left\{ v_i \geq p_i \right\} \right) r$
		\EndFor
		\State \textbf{return} $(p, \, r)$
	\end{algorithmic}\label{alg:computePrices}
\end{algorithm}

\paragraph{Finding Approximately-Optimal Prices.}
%
%
Algorithm~\ref{alg:computePrices} computes (in polynomial time) an approximately-optimal price vector for any non-Bayesian posted price auction.
%
It samples $K \in \mathbb{N}_{>0}$ matrices of buyers' valuations, each one drawn according to the distributions $\V$.
Then, it finds an optimal price vector $p$ in the discretized set $\mathcal{P}^b$,
assuming that buyers' valuations are drawn according to the empirical distribution resulting from the sampled matrices.\footnote{In this work, for a discretization step $b \in \mathbb{N}_{>0}$, we let $P^b \subset [0,1]$ be the set of prices multiples of $1/b$, while $\mathcal{P}^b \coloneqq \bigtimes_{i \in \N} P^b $.}
This last step can be done by backward induction, as it is well known in the literature (see, \emph{e.g.},~\citep{TaoComplexity2020}). 
%
%
%
The following Lemma~\ref{lm:samplingPublic} establishes the correctness of Algorithm~\ref{alg:computePrices}, also providing a bound on its running time.
The key ideas of its proof are: (i) the sampling procedure constructs a good estimation of the actual distributions of buyers' valuations; and (ii) even if the algorithm only considers discretized prices, the components of the computed price vector are at most $1/b$ less than those of an optimal (unconstrained) price vector.
As shown in the proof, this is strictly related to reducing buyer's valuations by $\frac{1}{b}$.
%
Thus, it follows by Lemma~\ref{lm:smallDecrease} that the seller's expected revenue is at most $1/b$ less than the optimal one.
\begin{restatable}{lemma}{lemmathree}\label{lm:samplingPublic}
	For any $\V = \{ \V_i \}_{i \in \N}$ and $\epsilon, \tau > 0$, there exist $K \in \textnormal{poly} \left( n, \frac{1}{\epsilon}, \log \frac{1}{\tau} \right)$ and $b \in \textnormal{poly} \left(  \frac{1}{\epsilon} \right)$ such that, with probability at least $1-\tau$, Algorithm~\ref{alg:computePrices} returns $(p,r)$ satisfying $\textsc{Rev}(\V,p)\geq \max_{p' \in [0,1]^n} \textsc{Rev}(\V,p')-\epsilon$ and $ r \in [\textsc{Rev}(\V,p)- \epsilon, \textsc{Rev}(\V,p)+\epsilon ]$ in time $\textnormal{poly}\left( n, \frac{1}{\epsilon}, \log\frac{1}{\tau} \right)$.
	%
	%
\end{restatable}

%


\section{Public Signaling}\label{sec:public}

In the following, we design a PTAS for computing a revenue-maximizing $(\phi,f)$ pair in the public signaling setting.
Notice that this positive result is tight by Theorem~\ref{thm:hardness}.
%

As a first intermediate result, we prove the compared stability of suitably-defined functions, which are intimately related to the seller's revenue.
In particular, for every price vector $p \in [0,1]^n$, we conveniently let $g_p: [0,1]^n \to [0,1]$ be a function that takes a vector of buyers' valuations and outputs the seller's expected revenue achieved by selecting $p$ when the buyers' valuations are those specified as input.
The following Lemma~\ref{lm:stabAuction} shows that, given some distributions of buyers' valuations $\V$ and a posterior $\xi \in \Delta_{\Theta}$, there always exists a price vector $p \in [0,1]^n$ such that $g_p$ is stable compared with $g_{p'}$ for every other $p' \in [0,1]^n$.
This result crucially allows us to decompose any posterior $\xi \in \Delta_{\Theta}$ by means of the decomposition lemma in Corollary~\ref{lm:decompositionweak}, while guaranteeing a small loss in terms of seller's expected revenue.
\begin{restatable}{lemma}{lemmafour}\label{lm:stabAuction}
	Given $\alpha, \epsilon>0$, a posterior $\xi \in \Delta_{\Theta}$, and some distributions of buyers' valuations $\V = \{ \V_i \}_{i \in \N}$, there exists $p \in [0,1]^n$ such that, for every other $p'\in [0,1]^n$, the function $g_p$ is $(1,\alpha,\epsilon)$-stable compared with $g_{p'}$ for $(\xi,\V)$.
\end{restatable}

Our PTAS leverages the fact that public signaling schemes can be represented as probability distributions over buyers' posteriors (recall that, in the public signaling setting, all the buyers share the same posterior, as they all observe the same signal).  
%
In particular, the algorithm returns a pair $(\gamma, f^\circ)$, where $\gamma$ is a probability distribution over $\Delta_{\Theta}$ satisfying consistency constraints (see Equation~\eqref{eq:w_consistent}), while $f^{\circ}: \Delta_{\Theta} \to [0,1]^n$ is a function mapping each posterior to a price vector.
%
In single-receiver settings, it is well known (see Subsection~\ref{subsec:signaling_auction}) that using distributions over posteriors rather than signaling schemes $\phi$ is without loss of generality. 
The following lemma shows that the same holds in our case, \emph{i.e.}, given a pair $(\gamma, f^\circ)$, it is always possible to obtain a pair $(\phi,f)$ providing the seller with the same expected revenue.
\begin{restatable}{lemma}{lemmafive}\label{lm:posteriorToSignaling}
	 Given a pair $(\gamma, f^\circ)$, where $\gamma$ is a probability distribution over $\Delta_{\Theta}$ with $\sum_{\xi \in \textnormal{supp}(\gamma)} \gamma(\xi) \xi(\theta) =\mu_\theta$ for all $\theta \in \Theta$ and $f^\circ : \Delta_{\Theta} \to [0,1]^n$, there is a pair $(\phi,f)$ s.t. 
	 \[
	 	\sum_{\theta \in \Theta} \hspace{-1mm} {\mu}_\theta \hspace{-1mm} \sum_{s \in \mathcal{S}} \hspace{-1mm} {\phi}_\theta(s) \textsc{Rev}(\mathcal{V},f(s), {\xi}_s) \hspace{-1mm}= \hspace{-5mm}\sum_{\xi \in \textnormal{supp}(\gamma)}\hspace{-4mm} \gamma(\xi) \textsc{Rev}(\V,f^\circ \hspace{-0.5mm} (\xi),\xi).
	 \]
\end{restatable}

Next, we show that, in order to find an approximately-optimal pair $(\gamma,f^\circ)$, we can restrict the attention to $q$-uniform posteriors (with $q$ suitably defined).
%
%
First, we introduce the following LP that computes an optimal probability distribution restricted over $q$-uniform posteriors.
\begin{subequations}\label{eq:lppublic}
	\begin{align}
	\max_{\gamma \in \Delta_{\Xi^q}} & \sum_{\xi \in \Xi^q}   \gamma(\xi) \max_{p \in [0,1]^n}  \textsc{Rev}(\V,p,\xi) \,\, \textnormal{s.t.} \\
	&\sum_{\xi \in \Xi^q} \gamma(\xi) \, \xi(\theta) =\mu_\theta & \forall \theta \in \Theta. \label{eq:consistency}
	\end{align}
\end{subequations}
The following Lemma~\ref{lm:optQUniform} shows the optimal value of LP~\ref{eq:lppublic} is ``close'' to $OPT$.
%
Its proof is based on the following core idea.
Given the signaling scheme $\phi$ in a revenue-maximizing pair $(\phi,f)$, letting $\gamma$ be the distribution over $\Delta_{\Theta}$ induced by $\phi$, we can decompose each posterior in the support of $\gamma$ according to Corollary~\ref{lm:decompositionweak}.
Then, the obtained distributions over $q$-uniform posteriors are consistent according to Equation~\eqref{eq:w_consistent}, and, thus, they satisfy Constraints~\eqref{eq:consistency}.
Moreover, since such distributions are also decreasing around the decomposed posteriors, by Lemma~\ref{lm:stabAuction} each time a posterior is decomposed there exists a price vector resulting in a small revenue loss.
These observations allow us to conclude that the seller's expected revenue provided by an optimal solution to LP~\ref{eq:lppublic} is within some small additive loss of $OPT$.

\begin{restatable}{lemma}{lemmasix}\label{lm:optQUniform}
	Given $\eta >0$ and letting $q = \frac{1}{\eta^2}128\log \frac{6}{\eta}$, an optimal solution to LP~\ref{eq:lppublic} has value at least $OPT-\eta$.
\end{restatable}


Finally, we are ready to provide our PTAS.
Its main idea is to solve LP~\ref{eq:lppublic} (of polynomial size) for the value of $q$ in Lemma~\ref{lm:optQUniform}.
%
%
This results in a small revenue loss.
The last part missing for the algorithm is computing the terms appearing in the objective of LP~\ref{eq:lppublic}, \emph{i.e.}, a revenue-maximizing price vector (together with its revenue) for every $q$-uniform posterior.
In order to do so, we can use Algorithm~\ref{alg:computePrices} (see also Lemma~\ref{lm:samplingPublic}), which allows us to obtain in polynomial time good approximations of such price vectors, with high probability. 
	
\begin{restatable}{theorem}{theoremthree}
	There exists an additive PTAS for computing a revenue-maximizing $(\phi,f)$ pair with public signaling.
	%
\end{restatable}

\section{Private Signaling}\label{sec:private}

With private signaling, computing a $(\phi,f)$ pair amounts to specifying a pair $(\phi_i,f_i)$ for each buyer $i \in \N$---composed by a marginal signaling scheme $\phi_i : \Theta \to \Delta_{\sset_i}$ and a price function $f_i : \sset_i \to [0,1]$ for buyer $i$---, and, then, correlating the $\phi_i$ so as to obtain a (non-marginal) signaling scheme $\phi: \Theta \to \Delta_\sset$.
We leverage this fact to design our PTAS.

%
In Subsection~\ref{sec:private_first}, we first show that it is possible to restrict the set of marginal signaling schemes of a given buyer $i \in \N $ to those encoded as distributions over $q$-uniform posteriors, as we did with public signaling.
Then, we provide an LP formulation for computing an approximately-optimal $(\phi,f)$ pair, dealing with the challenge of correlating marginal signaling schemes in a non-trivial way.
Finally, in Subsection~\ref{sec:private_second}, we show how to compute a solution to the LP in polynomial time, which requires the application of the ellipsoid method in a non-trivial way,
due to the features of the formulation.

\subsection{LP for Approximate Signaling Schemes}\label{sec:private_first}

%
%
Before providing the LP, we show that restricting marginal signaling schemes to $q$-uniform posteriors results in a buyer's behavior which is similar to the one with arbitrary posteriors.
This amounts to showing that suitably-defined functions related to the probability of buying are comparatively stable.

For $i \in \N$ and $p_i \in [0,1]$, let $g_{i,p_i}:[0,1]^n \to \{0,1\} $ be a function that takes as input a vector of buyers' valuations and outputs $1$ if and only if $v_i \geq p_i$ (otherwise it outputs $0$).
%
%
\begin{restatable}{lemma}{lemmaseven}\label{lm:singleBuyer}
	%
	Given $\alpha, \epsilon>0$ and some distributions $\V = \{ \V_i \}_{i \in \N}$, for every buyer $i \in \N$, posterior $\xi_i \in \Delta_{\Theta}$, and price $p_i \in [0,1]$, the function $g_{i, [p_i-\epsilon]_+}$ is $(0,\alpha,\epsilon)$-stable compared with $g_{i,p_i}$ for $(\xi_i,\V)$.
	%
\end{restatable}

The following remark will be crucial for proving Lemma~\ref{lm:privateQuniform}.
It shows that, if for every $i \in \N$ we decompose buyer $i$'s posterior $\xi_i \in \Delta_{\Theta}$ by means of a distribution over $q$-uniform posteriors $(\alpha,\epsilon)$-decreasing around $\xi_i$, then the probability with which buyer $i$ buys only decreases by a small amount.\footnote{In this section, for the ease of presentation, we abuse notation and use $\Xi_i^q$ to denote the (all equal) sets of $q$-uniform posteriors (Definition~\ref{def:s_distribution}), one per buyer $i \in \N$, while $\Xi^q \coloneqq \bigtimes_{i \in \N} \Xi^q_i$ is the set of tuples $\xi = (\xi_1,\ldots,\xi_n)$ specifying a $\xi_i \in \Xi_i^q$ for each $i \in \N$.}
\begin{remark}
	Lemma~\ref{lm:singleBuyer} and Theorem~\ref{thm:decomposition} imply that, given a tuple of posteriors $\xi = (\xi_1,\ldots,\xi_n) \in \bigtimes_{i \in \N} \Delta_\theta$ and some distributions $\V = \{ \V_i \}_{i \in \N}$, for every buyer $i \in \N$ and price $p_i \in [0,1]$, there exists $\gamma_i \in \Delta_{\Xi^q_i}$ with $q=\frac{32}{\epsilon^2} \log \frac{4}{\alpha} $ s.t.
	\[
	\underset{\tilde \xi_i \sim \gamma_i}{\mathbb{E}} \hspace{-1.2mm} \left[\hspace{-0.5mm} \tilde\xi_i(\theta) \hspace{-0.8mm} \underset{V \sim \V}{\pr} \hspace{-1.3mm} \left\{ \hspace{-0.5mm} V_i \tilde \xi_i \hspace{-0.5mm} \ge \hspace{-0.5mm} [p_i \hspace{-0.5mm} - \hspace{-0.5mm} \epsilon]_+\hspace{-0.6mm} \right\} \hspace{-0.8mm} \right] \hspace{-1mm} \ge \hspace{-0.5mm} \xi_i(\theta) (1-\alpha) \hspace{-0.5mm} \underset{V \sim \V }{\pr} \hspace{-1mm} \left\{ \hspace{-0.5mm}V_i \xi_i \hspace{-0.5mm} \ge \hspace{-0.5mm} p_i \hspace{-0.3mm} \right\}
	\]
	and $\sum_{\tilde\xi_i \in \Xi^q_i} \gamma_i(\tilde\xi_i) \, \tilde\xi_i(\theta)= \xi_i(\theta)$ for all $\theta \in \Theta$.
\end{remark}

Next, we show that an approximately-optimal pair $(\phi,f)$ can be found by solving LP~\ref{lp:private} instantiated with suitably-defined $q \in \mathbb{N}_{>0}$ and $b \in \mathbb{N}_{>0}$.
%
%
%
LP~\ref{lp:private} employs:
\begin{itemize}
	\item Variables $\gamma_{i, \xi_i}$ (for $i \in \N$ and $\xi_i \in \Xi^q_i$), which encode the distributions over posteriors representing the marginal signaling schemes $\phi_i : \Theta \to \Delta_{\sset_i}$ of the buyers.
	\item Variables $t_{i,\xi_i,p_i}$ (for $i \in \N$, $\xi_i \in \Xi^q_i$, and $p_i \in P^b$), with $t_{i,\xi_i,p_i}$ encoding the probability that the seller offers price $p_i$ to buyer $i$ and buyer $i$'s posterior is $\xi_i$.
	\item Variables $y_{\theta,\xi,p}$ (for $\theta \in \Theta$, $\xi \in \Xi^q$, and $p \in \mathcal{P}^b$), with $y_{\theta,\xi,p}$ encoding the probability that the state is $\theta$, the buyers' posteriors are those specified by $\xi$, and the prices that the seller offers to the buyers are those given by $p$.
\end{itemize}
\begin{subequations} \label{lp:private}
	\begin{align}
	\max_{\substack{\gamma, t, y\ge 0}}
	& \,\, \sum_{\theta \in \Theta} \sum_{\xi \in \Xi^q} \sum_{p \in \mathcal{P}^b} y_{\theta, \xi,p} \, \textsc{Rev}(\V,p,\xi) \quad \textnormal{s.t.} \\
	& \xi_i(\theta) t_{i,\xi_i,p_i} = \sum_{\xi' \in \Xi^q : \xi'_i = \xi_i} \sum_{p' \in \mathcal{P}^b : p'_i = p_i}  y_{\theta,\xi',p'} \nonumber \\
	& \hspace{1.2cm}\forall	\theta \in \Theta, \forall i \in \N, \forall \xi_i \in \Xi^q_i, \forall p_i \in P^b  \label{eq:private1} \\
	&\sum_{p_i \in P^b} t_{i,\xi_i,p_i} = \gamma_{i,\xi_i} \hspace{0.95cm}\forall i \in \N, \forall \xi_i \in \Xi^q_i \label{eq:private2}  \\
	& \sum_{\xi_i \in \Xi^q_i} \gamma_{i , \xi_i} \, \xi_i(\theta) =\mu_\theta \hspace{0.9cm} \forall i \in \N, \forall \theta \in \Theta. \label{eq:private3} 
	\end{align}
\end{subequations}
Variables $t_{i,\xi_i,p_i}$ represent marginal signaling schemes, allowing for multiple signals inducing the same posterior.
This is needed since signals may correspond to different price proposals.\footnote{Notice that, in a classical setting in which the sender does \emph{not} have to propose a price (or, in general, select some action after sending signals), there always exists a signaling scheme with no pair of signals inducing the same posterior. Indeed, two signals that induce the same posterior can always be joined into a single signal. This is \emph{not} the case in our setting, where we can only join signals that induce the same posterior and correspond to the same price.}
One may think of marginal signaling schemes in LP~\ref{lp:private} as if they were using signals defined as pairs $s_i = (\xi_i,p_i)$, with the convention that $f_i(s_i) = p_i$.
Variables $y_{\theta,\xi,p}$ and Constraints~\eqref{eq:private1} ensure that marginal signaling schemes are correctly correlated together, by directly working in the domain of the distributions over posteriors.

To show that an optimal solution to LP~\ref{lp:private} provides an approximately-optimal $(\phi,f)$ pair, we need the following two lemmas.
Lemma~\ref{lm:lpToSignaling} proves that, given a feasible solution to LP~\ref{lp:private}, we can recover a pair $(\phi,f)$ providing the seller with an expected revenue equal to the value of the LP solution.
Lemma~\ref{lm:privateQuniform} shows that the optimal value of LP~\ref{lp:private} is ``close'' to $OPT$.
%
These two lemmas imply that an approximately-optimal $(\phi,f)$ pair can be computed by solving LP~\ref{lp:private}.

\begin{restatable}{lemma}{lemmaeight} \label{lm:lpToSignaling}
	Given a feasible solution to LP~\ref{lp:private}, it is possible to recover a  pair $(\phi,f)$ that provides the seller with an expected revenue equal to the value of the solution.
\end{restatable}

\begin{restatable}{lemma}{lemmanine}\label{lm:privateQuniform}
	For every $\eta > 0$, there exist $b(\eta), q(\eta) \in \mathbb{N}_{>0}$ such that LP~\ref{lp:private} has optimal value at least $OPT-\eta$.
\end{restatable}

\subsection{PTAS}\label{sec:private_second}

%
We provide an algorithm that approximately solves LP~\ref{lp:private} in polynomial time, which completes our PTAS for computing a revenue-maximizing pair $(\phi,f)$ in the private setting.
The core idea of our algorithm is to apply the ellipsoid method on the dual of LP~\ref{lp:private}.\footnote{To be precise, we apply the ellipsoid method to the dual of a relaxed version of LP~\ref{lp:private}, since we need an over-constrained dual. More details on these technicalities are in the Extended Version.}
%
In particular, our implementation of the ellipsoid algorithm uses an approximate separation oracle that needs to solve the following optimization problem.
\begin{definition} [MAX-LINREV] \label{def:oracle}
	Given some distributions of buyers' valuations $\V = \{ \V_i \}_{i \in \N}$ such that each $\V_i$ has finite support and a vector $w \in [0,1]^{n \times |\Xi^q_i| \times |P^b|}$, solve
	\[
		\max_{\xi \in  \Xi^q, p \in \mathcal{P}^b} \textsc{Rev}(\V,p,\xi) + \sum_{i \in \N} w_{i,\xi_i,p_i}.
	\]
\end{definition}

As a first step, we provide an FPTAS for MAX-LINREV using a dynamic programming approach.
This will be the main building block of our approximate separation oracle.\footnote{Notice that, since MAX-LINREV takes as input distributions with a \emph{finite support}, we can safely assume that such distributions can be explicitly represented in memory. In our PTAS, the inputs to the dynamic programming algorithm are obtained by building empirical distributions through samples from the actual distributions of buyers' valuations, thus ensuring finiteness of the supports.}

The FPTAS works as follows.
Given an error tolerance $\delta > 0$, it first defines a step size $\frac{1}{c}$, with $c=\lceil \frac{n}{\delta}\rceil$, and builds a set $A=\{ 0,\frac{1}{c},\frac{2}{c},\dots,n \}$ of possible discretized values for the linear term appearing in the MAX-LINREV objective.
Then, for every buyer $i \in \N$ (in reversed order) and value $a \in A$, the algorithm computes $M(i,a)$, which is an approximation of the largest seller's revenue provided by a pair $(\xi,p)$ when considering buyers $i,\ldots,n$ only, and restricted to pairs $(\xi,p)$ such that the inequality $\sum_{j \in \N: j\ge i} w_{j,\xi_j,p_j}\ge a$ is satisfied.
By letting $z_i \coloneqq \Pr_{v_i \sim  \V_i} \left\{ v_i^\top \xi_i \ge p_i \right\}$, the value $M(i,a)$ can be defined by the following recursive formula:\footnote{Notice that, given a pair $(\xi,p)$ with $\xi \in \Xi^q$ and $p \in \mathcal{P}^b$, it is possible to compute in polynomial time the probability with which a buyer $i \in \N$ buys the item.}
\begin{align*}
	M(i,a) \coloneqq \max_{\substack{\xi_i \in \Xi^q_i, p_i \in P^b \\ a' \in A : w_{i,\xi_i,p_i}+a' \ge a}}  z_i p_i + \left( 1-z_i \right) M(i+1,a') .
	%
\end{align*}
 Finally, the algorithm returns $\max_{a \in A} \left\{ M(1,a)+a \right\}$.
 Thus:
 
\begin{restatable}{lemma}{lemmaeleven}
	  For any $\delta>0$, there exists a dynamic programming algorithm that provides a $\delta$-approximation (in the additive sense) to \emph{MAX-LINREV}.
	  Moreover, the algorithm runs in time polynomial in the size of the input and $\frac{1}{\delta}$.
\end{restatable}

Now, we are ready to prove the main result of this section.
\begin{restatable}{theorem}{theoremfour}
	There exists an additive PTAS for computing a revenue-maximizing $(\phi,f)$ pair with private signaling.
	%
\end{restatable}

\section*{Acknowledgments}	
This work has been partially supported by the Italian MIUR PRIN 2017 Project ALGADIMAR ``Algorithms, Games, and Digital Market''.

\bibliography{bibliography}

\begin{thebibliography}{29}
\providecommand{\natexlab}[1]{#1}

\bibitem[{Adamczyk et~al.(2017)Adamczyk, Borodin, Ferraioli, Keijzer, and
  Leonardi}]{adamczyk2017sequential}
Adamczyk, M.; Borodin, A.; Ferraioli, D.; Keijzer, B.~D.; and Leonardi, S.
  2017.
\newblock Sequential posted-price mechanisms with correlated valuations.
\newblock \emph{ACM Transactions on Economics and Computation (TEAC)}, 5(4):
  1--39.

\bibitem[{Babaioff et~al.(2017)Babaioff, Blumrosen, Dughmi, and
  Singer}]{babaioff2017posting}
Babaioff, M.; Blumrosen, L.; Dughmi, S.; and Singer, Y. 2017.
\newblock Posting prices with unknown distributions.
\newblock \emph{ACM Transactions on Economics and Computation (TEAC)}, 5(2):
  1--20.

\bibitem[{Babaioff et~al.(2015)Babaioff, Dughmi, Kleinberg, and
  Slivkins}]{babaioff2015dynamic}
Babaioff, M.; Dughmi, S.; Kleinberg, R.; and Slivkins, A. 2015.
\newblock Dynamic pricing with limited supply.
\newblock \emph{ACM Transactions on Economics and Computation (TEAC)}, 3(1):
  1--26.

\bibitem[{Bacchiocchi et~al.(2022)Bacchiocchi, Castiglioni, Marchesi, Romano,
  and Gatti}]{bacchio2022}
Bacchiocchi, F.; Castiglioni, M.; Marchesi, A.; Romano, G.; and Gatti, N. 2022.
\newblock Public Signaling in Bayesian Ad Auctions.
\newblock \emph{CoRR}, abs/2201.09728.

\bibitem[{Badanidiyuru, Bhawalkar, and Xu(2018)}]{badanidiyuru2018targeting}
Badanidiyuru, A.; Bhawalkar, K.; and Xu, H. 2018.
\newblock Targeting and Signaling in Ad Auctions.
\newblock In \emph{Proceedings of the Twenty-Ninth Annual ACM-SIAM Symposium on
  Discrete Algorithms}, SODA '18, 2545–2563. USA: Society for Industrial and
  Applied Mathematics.
\newblock ISBN 9781611975031.

\bibitem[{Candogan(2019)}]{candogan2019persuasion}
Candogan, O. 2019.
\newblock Persuasion in networks: Public signals and k-cores.
\newblock In \emph{Proceedings of the 2019 ACM Conference on Economics and
  Computation}, 133--134.

\bibitem[{Castiglioni, Celli, and
  Gatti(2020{\natexlab{a}})}]{castiglioni2019persuading}
Castiglioni, M.; Celli, A.; and Gatti, N. 2020{\natexlab{a}}.
\newblock Persuading Voters: It's Easy to Whisper, It's Hard to Speak Loud.
\newblock In \emph{The Thirty-Fourth AAAI Conference on Artificial
  Intelligence}, 1870--1877.

\bibitem[{Castiglioni, Celli, and
  Gatti(2020{\natexlab{b}})}]{castiglioni2020Public}
Castiglioni, M.; Celli, A.; and Gatti, N. 2020{\natexlab{b}}.
\newblock Public Bayesian Persuasion: Being Almost Optimal and Almost
  Persuasive.
\newblock \emph{CoRR}, abs/2002.05156.

\bibitem[{Castiglioni et~al.(2020)Castiglioni, Celli, Marchesi, and
  Gatti}]{castiglioni2020online}
Castiglioni, M.; Celli, A.; Marchesi, A.; and Gatti, N. 2020.
\newblock Online Bayesian Persuasion.
\newblock In Larochelle, H.; Ranzato, M.; Hadsell, R.; Balcan, M.~F.; and Lin,
  H., eds., \emph{Advances in Neural Information Processing Systems},
  volume~33, 16188--16198. Curran Associates, Inc.

\bibitem[{Castiglioni et~al.(2021{\natexlab{a}})Castiglioni, Celli, Marchesi,
  and Gatti}]{castiglioni2020signaling}
Castiglioni, M.; Celli, A.; Marchesi, A.; and Gatti, N. 2021{\natexlab{a}}.
\newblock Signaling in Bayesian Network Congestion Games: the Subtle Power of
  Symmetry.
\newblock In \emph{The Thirty-Fifth AAAI Conference on Artificial
  Intelligence}.

\bibitem[{Castiglioni and Gatti(2021)}]{CastiglioniPersuading2021}
Castiglioni, M.; and Gatti, N. 2021.
\newblock Persuading Voters in District-based Elections.
\newblock In \emph{Thirty-Fifth {AAAI} Conference on Artificial Intelligence,
  {AAAI} 2021, Thirty-Third Conference on Innovative Applications of Artificial
  Intelligence, {IAAI} 2021, The Eleventh Symposium on Educational Advances in
  Artificial Intelligence, {EAAI} 2021, Virtual Event, February 2-9, 2021},
  5244--5251. {AAAI} Press.

\bibitem[{Castiglioni et~al.(2021{\natexlab{b}})Castiglioni, Marchesi, Celli,
  and Gatti}]{castiglioni2021multi}
Castiglioni, M.; Marchesi, A.; Celli, A.; and Gatti, N. 2021{\natexlab{b}}.
\newblock Multi-Receiver Online Bayesian Persuasion.
\newblock In Meila, M.; and Zhang, T., eds., \emph{Proceedings of the 38th
  International Conference on Machine Learning}, volume 139 of
  \emph{Proceedings of Machine Learning Research}, 1314--1323. PMLR.

\bibitem[{Castiglioni, Marchesi, and Gatti(2022)}]{castiglioni2022bayesian}
Castiglioni, M.; Marchesi, A.; and Gatti, N. 2022.
\newblock Bayesian Persuasion Meets Mechanism Design: Going Beyond
  Intractability with Type Reporting.
\newblock arXiv:2202.00605.

\bibitem[{Chawla et~al.(2010)Chawla, Hartline, Malec, and
  Sivan}]{chawla2010multi}
Chawla, S.; Hartline, J.~D.; Malec, D.~L.; and Sivan, B. 2010.
\newblock Multi-parameter mechanism design and sequential posted pricing.
\newblock In \emph{Proceedings of the forty-second ACM symposium on Theory of
  computing}, 311--320.

\bibitem[{Cheng et~al.(2015)Cheng, Cheung, Dughmi, Emamjomeh{-}Zadeh, Han, and
  Teng}]{ChengMixture2015}
Cheng, Y.; Cheung, H.~Y.; Dughmi, S.; Emamjomeh{-}Zadeh, E.; Han, L.; and Teng,
  S. 2015.
\newblock Mixture Selection, Mechanism Design, and Signaling.
\newblock In Guruswami, V., ed., \emph{{IEEE} 56th Annual Symposium on
  Foundations of Computer Science, {FOCS} 2015, Berkeley, CA, USA, 17-20
  October, 2015}, 1426--1445. {IEEE} Computer Society.

\bibitem[{Correa et~al.(2017)Correa, Foncea, Hoeksma, Oosterwijk, and
  Vredeveld}]{correa2017posted}
Correa, J.; Foncea, P.; Hoeksma, R.; Oosterwijk, T.; and Vredeveld, T. 2017.
\newblock Posted price mechanisms for a random stream of customers.
\newblock In \emph{Proceedings of the 2017 ACM Conference on Economics and
  Computation}, 169--186.

\bibitem[{Einav et~al.(2018)Einav, Farronato, Levin, and
  Sundaresan}]{einav2018auctions}
Einav, L.; Farronato, C.; Levin, J.; and Sundaresan, N. 2018.
\newblock Auctions versus posted prices in online markets.
\newblock \emph{Journal of Political Economy}, 126(1): 178--215.

\bibitem[{eMarketer(2021)}]{eMarketer}
eMarketer. 2021.
\newblock Worldwide ecommerce will approach \$5 trillion this year.
\newblock
  \url{https://www.emarketer.com/content/worldwide-ecommerce-will-approach-5-trillion-this-year}.
\newblock Accessed: 2021-09-02.

\bibitem[{Emek et~al.(2014)Emek, Feldman, Gamzu, PaesLeme, and
  Tennenholtz}]{emek2014signaling}
Emek, Y.; Feldman, M.; Gamzu, I.; PaesLeme, R.; and Tennenholtz, M. 2014.
\newblock Signaling schemes for revenue maximization.
\newblock \emph{ACM Transactions on Economics and Computation}, 2(2): 1--19.

\bibitem[{Kamenica and Gentzkow(2011)}]{kamenica2011bayesian}
Kamenica, E.; and Gentzkow, M. 2011.
\newblock Bayesian persuasion.
\newblock \emph{American Economic Review}, 101(6): 2590--2615.

\bibitem[{Khot and Saket(2012)}]{KhotHarness2013}
Khot, S.; and Saket, R. 2012.
\newblock Hardness of Finding Independent Sets in Almost q-Colorable Graphs.
\newblock In \emph{2013 IEEE 54th Annual Symposium on Foundations of Computer
  Science}, 380--389. Los Alamitos, CA, USA: IEEE Computer Society.

\bibitem[{Kleinberg and Leighton(2003)}]{kleinberg2003value}
Kleinberg, R.; and Leighton, T. 2003.
\newblock The value of knowing a demand curve: Bounds on regret for online
  posted-price auctions.
\newblock In \emph{44th Annual IEEE Symposium on Foundations of Computer
  Science, 2003. Proceedings.}, 594--605. IEEE.

\bibitem[{Rabinovich et~al.(2015)Rabinovich, Jiang, Jain, and
  Xu}]{rabinovich2015information}
Rabinovich, Z.; Jiang, A.~X.; Jain, M.; and Xu, H. 2015.
\newblock Information disclosure as a means to security.
\newblock In \emph{Proceedings of the 2015 International Conference on
  Autonomous Agents and Multiagent Systems}, 645--653.

\bibitem[{Romano et~al.(2021)Romano, Tartaglia, Marchesi, and
  Gatti}]{romano2021online}
Romano, G.; Tartaglia, G.; Marchesi, A.; and Gatti, N. 2021.
\newblock Online Posted Pricing with Unknown Time-Discounted Valuations.
\newblock In \emph{Thirty-Fifth {AAAI} Conference on Artificial Intelligence,
  {AAAI} 2021, Thirty-Third Conference on Innovative Applications of Artificial
  Intelligence, {IAAI} 2021, The Eleventh Symposium on Educational Advances in
  Artificial Intelligence, {EAAI} 2021, Virtual Event, February 2-9, 2021},
  5682--5689. {AAAI} Press.

\bibitem[{Seifert(2006)}]{seifert2006posted}
Seifert, S. 2006.
\newblock \emph{Posted price offers in internet auction markets}, volume 580.
\newblock Springer Science \& Business Media.

\bibitem[{Shah, Johari, and Blanchet(2019)}]{Shah2019Semi}
Shah, V.; Johari, R.; and Blanchet, J. 2019.
\newblock Semi-Parametric Dynamic Contextual Pricing.
\newblock In \emph{Advances in Neural Information Processing Systems},
  2363--2373.

\bibitem[{Vasserman, Feldman, and Hassidim(2015)}]{vasserman2015implementing}
Vasserman, S.; Feldman, M.; and Hassidim, A. 2015.
\newblock Implementing the wisdom of waze.
\newblock In \emph{Twenty-Fourth International Joint Conference on Artificial
  Intelligence}, 660--666.

\bibitem[{Xiao, Liu, and Huang(2020)}]{TaoComplexity2020}
Xiao, T.; Liu, Z.; and Huang, W. 2020.
\newblock On the Complexity of Sequential Posted Pricing.
\newblock In \emph{Proceedings of the 19th International Conference on
  Autonomous Agents and MultiAgent Systems}, AAMAS '20, 1521–1529. Richland,
  SC: International Foundation for Autonomous Agents and Multiagent Systems.
\newblock ISBN 9781450375184.

\bibitem[{Xu(2020)}]{XuTractability2020}
Xu, H. 2020.
\newblock On the Tractability of Public Persuasion with No Externalities.
\newblock In Chawla, S., ed., \emph{Proceedings of the 2020 {ACM-SIAM}
  Symposium on Discrete Algorithms, {SODA} 2020, Salt Lake City, UT, USA,
  January 5-8, 2020}, 2708--2727. {SIAM}.

\end{thebibliography}

\clearpage
\onecolumn
\appendix

\section{Additional Discussion on Lemma~\ref{lm:smallDecrease}}\label{app:lemma_discussion}

%
Our main contribution in Section~\ref{sec:non_bayesian} (Lemma~\ref{lm:smallDecrease}) is to show that: $\max_{p \in [0,1]^n} \textsc{Rev}(\V^\epsilon,p) \geq \max_{p \in [0,1]^n} \textsc{Rev}(\V,p) -\epsilon$.
By letting $p^* \in \argmax_{p \in [0,1]^n} \textsc{Rev}(\V,p)$ be a revenue-maximizing price vector for the seller under distributions $\V$, one may na\"ively think of proving the result by simply showing that, given the price vector $p^{*,\epsilon} \in [0,1]^n$ with $p^{*,\epsilon}_i \coloneqq [p_i^*-\epsilon]_{+}$ and distributions $\V^\epsilon$, each buyer would buy the item at least with the same probability as for price vector $p^*$ and distributions $\V$, but paying a price that is only $\epsilon$ less.
This would imply that $\textsc{Rev}(\V^\epsilon,p^{*,\epsilon}) \geq  \textsc{Rev}(\V,p^*)-\epsilon$, proving the desired result.
However, this line of reasoning does \emph{not} work, since, as shown by the following example, it has a major flaw. 

\begin{example} \label{ex:fail}
	Consider a posted price auction with two buyers.
	In the first case (distributions $\V$), buyer 1 has valuation $v_1=\frac{1}{2}$ and buyer 2 has valuation $v_2=1$.
	In such setting, an optimal price vector $p^*$ is such that $p_1^*= \frac{1}{2}+\epsilon$ and $p_2^* = 1$, so that the revenue of the seller, namely $\textsc{Rev}(\V,p^{*})$, is $1$.
	In the second case (distributions $\V^\epsilon$), buyer 1 has valuations $v_1=\frac{1}{2}$ and buyer 2 has valuation $v_2=1-\epsilon$.
	Thus, the revenue of the seller for the price vector $p^{*,\epsilon}$ (with $p^{*,\epsilon}_1 = \frac{1}{2}$ and $p^{*,\epsilon}_2 = 1 -\epsilon$), namely $\textsc{Rev}(\V^\epsilon,p^{*,\epsilon})$, is $\frac{1}{2}$, since buyer 1 will buy the item.
\end{example}

The crucial feature of the setting described in Example~\ref{ex:fail} is that there is an optimal price vector in which one buyer (buyer 1) is offered a price that is too low, and, thus, the seller prefers not to sell the item to them, but rather to another buyer (buyer 2).
This prevents a direct application of the line of reasoning outlined above.
However, one could circumvent this issue by considering an optimal price vector such that the seller is never upset if some buyer buys.
In other words, prices must be such that each buyer is proposed a price that is at least as large as the seller's expected revenue in the posted price auction restricted to buyers following them.
In Example~\ref{ex:fail}, the optimal price vector $p^*$ such that $p_1^* =p_2^* = 1$ would be fine.

\section{Proofs omitted from Section~\ref{sec:hardness}}\label{app:proof}

\hardness*

\begin{proof}
	We employ a reduction from an $\mathsf{NP}$-hard problem originally introduced by~\citet{KhotHarness2013}, which we formally state in the following.
	For any positive integer $k \in \mathbb{N}_{>0}$, integer $l \in \mathbb{N}$ such that $l \ge 2^k + 1$, and arbitrarily small constant $\epsilon > 0$, the problem reads as follows. Given an undirected graph $G \coloneqq (U,E)$, distinguish between:
	\begin{itemize}
		\item \emph{Case 1}. There exists a $l$-colorable induced subgraph of $G$ containing a $1 - \epsilon$ fraction of all vertices, where each color class contains a $\frac{1-\epsilon}{l}$ fraction of all vertices.\footnote{A \emph{$l$-colorable induced subgraph} is identified by a subset of vertices such that it is possible to assign one among $l$ different colors to each vertex, in such a way that there are \emph{no} two adjacent vertices having the same color. Given some color, its associated color class is the subset of all vertices in the subgraph having that color.}
		\item \emph{Case 2}. Every independent set of $G$ contains less than a $\frac{1}{l^{k+1}}$ fraction of all vertices.\footnote{An \emph{independent set} of $G$ is a subset of vertices such that there are \emph{no} two adjacent vertices.}
	\end{itemize}
	We reduce from such problem for $k=2, l=5$, and $\epsilon=\frac{1}{2}$.
	Our reduction works as follows:
	\begin{itemize}
		\item \emph{Completeness}. If \emph{Case 1} holds, then there exists a signaling scheme, price function pair $(\phi,f)$ that provides the seller with an expected revenue at least as large as some threshold $\eta$ (see Equation~\eqref{eq:eta_hardness} below for its definition).
		\item \emph{Soundness}. If \emph{Case 2} holds, then the seller's expected revenue for any signaling scheme, price function pair $(\phi,f)$ is smaller than $\eta-\delta$ with $\delta \coloneqq \frac{1}{m^2}$, where $m$ denotes the number or vertices of the graph $G$.
	\end{itemize}
	This shows that it is $\mathsf{NP}$-hard to approximate the optimal seller's expected revenue up to within an additive error $\delta$.
	Thus, since $\delta$ depends polynomially on the size of the problem instance, this also shows that there is no additive FPTAS for the problem of computing a revenue-maximizing $(\phi,f)$ pair, unless $\mathsf{P}=\mathsf{NP}$.
	%
	%
	%
	
	\paragraph{Construction}
	Given an undirected graph $G \coloneqq (U,E)$, with vertices $U \coloneqq \{u_1,\dots,u_m\}$, we build a single-buyer Bayesian posted price auction as follows.\footnote{In a single-buyer setting, we always omit the subscript $i$ from symbols, as it is clear that they refer to the unique buyer. Moreover, with an overload of notation, we use buyer's signals as if they were signal profiles.}
	There is one state of nature $\theta_u \in \Theta$ for each vertex $u \in U$, and the prior belief over states $\mu \in \Delta_{\Theta}$ is such that $\mu_{\theta_u}=\frac{1}{m}$ for all $u \in U$.
	There is a finite set of possible buyer's valuations.
	For every vertex $u \in U$, there is a valuation vector $v_u \in [0,1]^m$ such that:
	\begin{itemize}
		\item $v_u(\theta_u) = 1$;
		\item $v_u(\theta_{u'}) = \frac{1}{2}$ for all $u' \in U : (u,u') \notin E$; and 
		\item $v_u(\theta_{u'}) = 0$ for all $u' \in U : (u,u') \in E$.
	\end{itemize}
	Each valuation $v_u$ has probability $\frac{1}{m^2}$ of occurring according to the distribution $\V$.
	Moreover, there is an additional valuation vector $v_o \in [0,1]^m$ such that $v_o(\theta) = \frac{1}{2}+\frac{l}{(1-\epsilon) 2m} $ for all $\theta \in \Theta$, having probability $1 - \frac{1}{m}$.
	%
	%

	\paragraph{Completeness}
	%
	%
	Assume that a $l$-colored induced subgraph of $G$ is given, and that it contains a fraction $1-\epsilon $ of vertices, while each color class is made up of a fraction $\frac{1-\epsilon}{l} $ of all vertices.
	We let $L \coloneqq \{1,\dots,l\}$ be the set of possible colors, with $j \in L$ denoting a generic color.
	In the following, we show how to build a signaling scheme, price function pair $(\phi,f)$ that provides the seller with an expected revenue greater than or equal to a suitably-defined threshold $\eta$ (Equation~\eqref{eq:eta_hardness}).
	The seller has $l+1$ signals available, namely $\sset \coloneqq \{ s_j \}_{j \in L} \cup \{ s_o \}$.
	For every vertex $u \in U$, if $u$ has been assigned some color $j \in L$ (that is, $u$ belongs to the induced subgraph), then we set $\phi_{\theta_u}(s_j) = 1$ and $\phi_{\theta_u}(s) = 0$ for all $s \in \sset : s \neq s_j$; otherwise, if $u$ has no color (that is, $u$ does not belong to the given subgraph), then we set $\phi_{\theta_u}(s_o) = 1 $ and $\phi_{\theta_u}(s) = 0$ for all $s\in \sset: s \neq s_o$.
	Moreover, the price function is such that
	\[
	f(s) = p_o \coloneqq \frac{1}{2}+\frac{l}{(1-\epsilon) 2m}  \,\,\, \text{for every signal} \,\,\, s \in \sset.
	\]
	%
	%
	%
	Next, we prove that, after receiving a signal $s_j \in \sset$ associated with some color $j \in L$, if the buyer has valuation $v_u$ for a node $u \in U$ colored of color $j$, then they will buy the item.
	In particular, the buyer's posterior $\xi_{s_j} \in \Delta_{\Theta}$ induced by signal $s_j$ is such that only state $\theta_u$ and states $\theta_{u'}$ for $u' \in U: (u,u') \notin E$ have positive probability (since, when the seller sends signal $s_j$, it must be the case that the vertex corresponding to the actual state of nature is colored of color $j$).
	Moreover, such probabilities are equal to $\xi_{s_j}(\theta_u) = \xi_{s_j}(\theta_{u'}) = \frac{l}{(1-\epsilon)m} $ (by applying Equation~\eqref{eq:posterior} and using the fact that each color class has a fraction $\frac{1-\epsilon}{l}$ of vertices).
	Thus, since $v_u(\theta_u) = 1$ and $v_u(\theta_{u'}) = \frac{1}{2}$ for all $u' \in U: (u,u') \notin E$, the expected valuation of the buyer given the posterior $\xi_{s_j}$ is
	\[
	\sum_{\theta \in \Theta} v_u(\theta) \, \xi_{s_j}(\theta) = \frac{1}{2} \left[ 1-\frac{l}{(1-\epsilon) m} \right] + \frac{l}{(1-\epsilon) m}=\frac{1}{2}+\frac{l}{(1-\epsilon) 2m}= p_o,
	\]
	and the buyer will buy the item.
	%
	%
	%
	Furthermore, when the seller sends signal $s_o$, their expected revenue is at least $ (1-\frac{1}{m}) p_o$, as it is always the case that the buyer buys the item when they have valuation $v_o$.
	Since the total probability of sending signals $s_j \in \sset$ for $j \in L$ is $1-\epsilon $ (given that the subgraph contains a fraction $1-\epsilon$ of vertices) and the probability of sending signal $s_o$ is $\epsilon $, we have that the seller's expected revenue is at least
	\begin{equation}\label{eq:eta_hardness}
		\eta \coloneqq (1-\epsilon) \left[ \frac{1-\epsilon}{m \, l}+1-\frac{1}{m} \right] p_o+ \epsilon \left(1-\frac{1}{m} \right)  p_o = \left[ \frac{(1-\epsilon)^2}{m \, l}+ \left( 1-\frac{1}{m} \right) \right]  p_o.
	\end{equation}
	where the factor $\frac{1-\epsilon}{m \, l}+1-\frac{1}{m} $ represents the probability that the buyer buys when sending a signal $s_j$ (this happens when either the buyer has valuation $v_u$ for a vertex $u \in U$ colored of color $j$ or the buyer has valuation $v_o$).

	%

	\paragraph{Soundness}
	%
	%
	By contradiction, we show that, if there exists a signaling scheme, price function pair $(\phi,f)$ with seller's expected revenue exceeding $\eta - \delta$, then the graph $G$ admits an independent set of size $\frac{1}{2l} (1-\epsilon)^2m>\frac{m}{l^{k+1}}$ (recall the choice of values for $k$, $l$, and $\epsilon$).
	If the seller's revenue is greater than $\eta-\delta$, by an averaging argument there must be at least one signal $s^* \in \sset$ whose contribution to the revenue $\sum_{\theta \in \Theta} {\mu}_\theta {\phi}_\theta(s^*) \textsc{Rev} \left( \mathcal{V},f(s^*), {\xi}_{s^*} \right)$ is more than $\eta-\delta$, where $\xi_{s^*} \in \Delta_{\Theta}$ is the buyer's posterior induced by signal $s^*$.
	%
	%
	Since the expected revenue cannot exceed the expected payment, the price $p^*\coloneqq f(s^*)$ that the seller proposes to the buyer when signal $s^*$ is sent must be greater than
	\begin{align*}
		\eta-\delta &=\left[\frac{(1-\epsilon)^2}{m \, l}+ \left(1-\frac{1}{m} \right)\right] p_o -\delta\\
		& = \left[\frac{(1-\epsilon)^2}{m \, l}+ \left(1-\frac{1}{m} \right)\right]  \left[\frac{1}{2}+\frac{l}{(1-\epsilon) 2m}\right]-\delta \\
		& \geq \left(1-\frac{1}{m} \right)\left[\frac{1}{2}+\frac{l}{(1-\epsilon) 2m}\right]-\delta  \\
		&\geq \frac{1}{2}+\frac{l}{(1-\epsilon)2m}-\frac{1}{m}-\frac{l}{(1-\epsilon) 2m^2}-\delta > \frac{1}{2},
	\end{align*} 
	where the last inequality holds for $m \geq 2$ since we set $\epsilon = \frac{1}{2}$, $l =5$, and $\delta = \frac{1}{m^2}$.
	Additionally, the price $p^*$ must be smaller than $p_o$ (see the completeness proof), otherwise, when the buyer has valuation $v_o$, they would never buy the item, resulting in a contribution to the seller's revenue at most of $\frac{1}{m}$ (recall that $v_o$ happens with probability $1-\frac{1}{m}$ and all other buyer's valuations do not exceed $1$).
	As a result, it must be the case that $p^* \in \left( \frac{1}{2}, p_o \right]$.
	Next, we prove that, after receiving signal $s^*$, the buyer will buy the item in all the cases in which their valuation belongs to a subset of valuations $v_u$ containing at least a fraction $\frac{1}{2l}(1-\epsilon)^2$ of all the valuations $v_u$.
	Indeed, if this is not the case, then the contribution to the seller's expected revenue due to signal $s^*$ would be less than
	\begin{align*}
		p^*\left[1-\frac{1}{m}+\frac{(1-\epsilon)^2}{2m \, l}\right]&\le  p_o \left[1-\frac{1}{m}+\frac{(1-\epsilon)^2}{2m \, l}\right] =  \eta - p_o \frac{(1-\epsilon)^2}{2m \, l} \le \eta-\delta,
	\end{align*}
	where the last inequality holds since $\delta=\frac{1}{m^2}\leq  p_o\frac{(1-\epsilon)^2}{2m \, l}$ for $m$ large enough. 
	%
	%
	%
	Let $U^* \subseteq U$ be the subset of vertices $u \in U$ such that the buyer will buy the item for their corresponding valuations $v_u$.
	We have that $|U^*| \geq \frac{m}{2l}(1-\epsilon)^2$.
	Next, we show that $U^*$ constitutes an independent set of $G$.
	First, since $p^* > \frac{1}{2}$, when the buyer's valuation is $v_u$ such that $u \in U^*$, then the buyer must value the item more than $\frac{1}{2}$, otherwise they would not buy.
	By contradiction, suppose that there is a couple of vertices $u,u' \in U^*$ such that $(u,u') \in E$.
	%
	W.l.o.g., let us assume that $\xi_{s^*}(\theta_u) \leq \xi_{s^*}(\theta_{u'})$.
	Then, the buyer's expected valuation induced by posterior $\xi_{s^*}$ is
	\[
	\xi_{s^*}(\theta_u)+\frac{1}{2} \left( 1 -\xi_{s^*}(\theta_u)-\xi_{s^*}(\theta_{u'}) \right)= \frac{1+\xi_{s^*}(\theta_u)-\xi_{s^*}(\theta_{u'})}{2}\le\frac{1}{2},
	\]
	which is a contradiction.
	Given that, by our initial assumption, the size of every independent set must be smaller than $\frac{m}{l^{k+1}}<\frac{1}{2l}(1-\epsilon)^2m$, we reach the final contradiction proving the result.
	%
	%
	%
\end{proof}

%
%
%
%

%
%
%
%
%


\section{Proofs omitted from Section~\ref{sec:general}}

\decomposition*

\begin{proof} 
	The probability distribution $\gamma \in \Delta_{\Xi^q}$ over $q$-uniform posteriors in the statement is defined as follows.
	Let $\xi^q \in \Xi^q$ be a buyer's posterior defined as the empirical mean of $q$ vectors built form $q$ {i.i.d.} samples drawn from the given posterior $\xi$.
	In particular, each sample is obtained by randomly drawing a state of nature, with each state $\theta \in \Theta$ having probability $\xi(\theta)$ of being selected, and, then, a $d$-dimensional vector is built by letting all its components equal to $0$, except for that one corresponding to $\theta$, which is set to $1$.
	Notice that $\xi^q$ is a random vector supported on $q$-uniform posteriors, whose expected value is posterior $\xi$.
	Then, $\gamma$ is such that, for every $\tilde \xi \in \Xi^q$, it holds $\gamma(\tilde \xi) = \textnormal{Pr} \left\{  \xi^q = \tilde \xi \right\}$.
	
	It is easy to check that $\xi(\theta) =\sum_{\tilde\xi \in \Xi^q} \gamma(\tilde\xi) \, \tilde\xi(\theta) = \mathbb{E}_{\tilde\xi \sim \gamma} \left[ \tilde\xi(\theta) \right]$ for all $\theta \in \Theta$, proving the first condition needed.

	Next, we prove that $\gamma$ satisfies Equation~\eqref{eq:general_bound}.
	To do so, we first introduce some useful definitions.
	For every $q$-uniform posterior $\tilde\xi \in \Xi^q$, with an overload of notation we let $\gamma(\tilde\xi,\theta,j)$ be the conditional probability of having drawn $\tilde\xi$ from $\gamma$ given that the drawn posterior assigns probability $\frac{j}{q}$ to state $\theta \in \Theta$ with $j \in \{0, \ldots, q\}$.
	Formally, for every $\tilde \xi \in \Xi^q$:
	\begin{align*}
		\gamma(\tilde\xi,\theta,j) \coloneqq
		\begin{cases}
			\displaystyle\frac{ \gamma(\tilde\xi)}{\sum_{\xi'\in \Xi^q:\xi'(\theta)=j/q}\gamma(\xi')} & \text{if $\tilde\xi(\theta)=\frac{j}{q}$}\\
			& \\
			\hspace{1cm}0 & \text{otherwise}
		\end{cases}.
	\end{align*}
	Then, for every $\theta \in \Theta$ and $j \in \{0,\ldots, q\}$, we let $\gamma^{\theta,j} $ be a probability distribution over $\Delta_{\Theta}$ supported on $q$-uniform posteriors such that $\gamma^{\theta,j}(\tilde\xi) \coloneqq \gamma(\tilde\xi,\theta,j)$ for all $\tilde\xi \in \Xi^q$.
	Moreover, for every buyer $i \in \N$ and matrix $V \in [0,1]^{n \times d}$ of buyers' valuations, we let $\Xi^{i,V} \subseteq \Xi^q$ be the set of  $q$-uniform posteriors that do \emph{not} change buyer $i$'s expected valuation by more than an additive factor $\epsilon$ with respect to their valuation in posterior $\xi$.
	Formally, 
	\[
	\Xi^{i,V} \coloneqq \left\{ \tilde\xi \in \Xi^q \mid  \left| V_i\xi-V_i\tilde\xi \right| \le \epsilon \right\}.
	\]

	In order to complete the proof, we introduce the following three lemmas (with Lemmas~\ref{lm:small}~and~\ref{lm:small_2} being adapted from~\citep{CastiglioniPersuading2021}).
	The first lemma shows that, for every state of nature $\theta \in \Theta$, it is possible to bound the cumulative probability mass that the distribution $\gamma$ assigns to $q$-uniform posteriors $\tilde\xi \in \Xi^q$ such that $\tilde\xi(\theta)$ differs from $\xi(\theta)$ by at least $\frac{\epsilon}{4}$ (in absolute terms).
	Formally:
	\begin{lemma}[Essentially Lemma~5 by~\citet{CastiglioniPersuading2021}]\label{lm:small}
		Given $\xi \in \Delta_{\Theta}$, for every $\theta \in \Theta$ it holds:
		\[
		\sum_{j  : \left| \frac{j}{q}-\xi(\theta) \right| \ge \frac{\epsilon}{4}} \; \sum_{\tilde\xi \in \Xi^q: \hat\xi(\theta)=\frac{j}{q}} \gamma(\tilde\xi) \le \frac{\alpha}{2} \, \xi(\theta),
		\]
		where $\gamma \in \Delta_{\Xi^q}$ is the probability distribution over $q$-uniform posteriors introduced at the beginning of the proof.
	\end{lemma}
	The second lemma, which is useful to prove Lemma~\ref{lm:stabilty}, shows that, for $q$-uniform posteriors $\tilde \xi \in \Xi^q$ such that $\tilde\xi (\theta)$ is sufficiently close to $\xi(\theta)$ for a state of nature $\theta \in \Theta$, the expected utility of each buyer is close to their utility in the given posterior $\xi$ with high probability.
	Formally:
	\begin{lemma}[Essentially Lemma~6 by~\citet{CastiglioniPersuading2021}]\label{lm:small_2}
		Given $\xi  \in \Delta_{\Theta}$, matrix $V \in [0,1]^{n \times d}$ of buyers' valuations, state of nature $\theta \in \Theta$, and $j : \left| \frac{j}{q}-\xi(\theta) \right| \le \frac{\epsilon}{4}$, the following holds for every buyer $i \in \N$:
		\[
		\sum_{\tilde\xi \in \Xi^{i,V} : \tilde\xi(\theta)=\frac{j}{q} } \gamma(\tilde\xi) \ge \left(1-\frac{\alpha}{2}\right) \sum_{\tilde\xi \in \Xi^q : \tilde\xi(\theta)=\frac{j}{q} } \gamma(\tilde\xi),
		\]
		where $\gamma \in \Delta_{\Xi^q}$ is the probability distribution over $q$-uniform posteriors introduced at the beginning of the proof.
	\end{lemma}
	Finally, the third lemma that we need reads as follows:
	\begin{lemma}\label{lm:stabilty}
		Given $\xi \in \Delta_{\Theta}$, for every state of nature $\theta \in \Theta$ and $j:\left| \frac{j}{q}-\xi(\theta) \right| \le \frac{\epsilon}{4}$, the probability distribution $\gamma^{\theta, j}$ defined at the beginning of the proof is $\left( \frac{\alpha}{2},\epsilon \right)$-decreasing around posterior $\xi$.
	\end{lemma}
	\begin{proof}
		According to Definition~\ref{def:perturbation} and by reversing the inequalities, we need to prove that, for every matrix $V \in [0,1]^{n \times d}$ of buyers' valuations and buyer $i \in \N$, it holds $\textnormal{Pr}_{\tilde\xi \sim \gamma^{\theta, j}} \left\{ V_i\tilde\xi \geq V_i\xi -\epsilon \right\} \geq 1 -\frac{\alpha}{2}$.
		By using the definition of the set $\Xi^{i, V}$, Lemma~\ref{lm:small_2}, and the definition of $\gamma^{\theta,j}$, we can write the following:
		\begin{align*}
			\textnormal{Pr}_{\tilde\xi \sim \gamma^{\theta, j}} \left\{ V_i\tilde\xi \geq V_i\xi-\epsilon \right\} & = \sum_{\tilde\xi \in \Xi^{i,V}} \gamma^{\theta,j}(\tilde\xi)=  \sum_{\tilde\xi \in \Xi^{i,V}: \tilde\xi(\theta) = \frac{j}{q}} \frac{\gamma(\tilde\xi)}{\sum_{\xi' \in \Xi^q: \xi'(\theta)=\frac{j}{q}} \gamma(\xi') } \\
			& \ge \left(  1- \frac{\alpha}{2} \right) \sum_{\tilde\xi \in \Xi^q: \tilde\xi(\theta) = \frac{j}{q}} \frac{\gamma(\tilde\xi)}{\sum_{\xi' \in \Xi^q: \xi'(\theta) = \frac{j}{q}} \gamma(\xi') } = 1 - \frac{\alpha}{2},
		\end{align*}
		which proves the lemma.
	\end{proof}

	Now, we are ready to prove the theorem, by means of the following inequalities:
	%
	\begin{align*}
		\mathbb{E}_{\tilde\xi \sim \gamma, V \sim \V} \Big[ \tilde\xi(\theta) \, g( V\tilde\xi ) \Big] & = \sum_{\tilde\xi\in \Xi^q} \gamma(\tilde\xi) \, \tilde\xi(\theta) \mathbb{E}_{V \sim \V} \Big[ g( V\tilde\xi  ) \Big]\\
		& \hspace{-2cm}\geq  \sum_{j: \left| \frac{j}{q}-\xi(\theta) \right| \leq \frac{\epsilon}{4}} \frac{j}{q} \sum_{\tilde\xi \in \Xi^q: \tilde\xi(\theta)=\frac{j}{q}} \gamma(\tilde\xi)  \, \mathbb{E}_{V \sim \V} \Big[ g( V\tilde\xi  ) \Big] \quad \textnormal{(By dropping terms from the sum)} \\
		& \hspace{-2cm} = \sum_{j: \left| \frac{j}{q}-\xi(\theta) \right| \leq \frac{\epsilon}{4}} \frac{j}{q}  \left( \sum_{\xi' \in \Xi^q:\xi'(\theta)=\frac{j}{q}} \gamma(\xi')  \right) \sum_{\tilde\xi \in \Xi^q:\tilde\xi(\theta)=\frac{j}{q}} \frac{\gamma(\tilde\xi)  }{\sum_{\xi' \in\Xi^q:	\xi'(\theta)=\frac{j}{q}} \gamma(\xi')} \mathbb{E}_{V \sim \V} \Big[ g( V\tilde\xi ) \Big]  \\
		& \hspace{-2cm} = \sum_{j: \left| \frac{j}{q}-\xi(\theta) \right| \leq \frac{\epsilon}{4}} \frac{j}{q}  \left( \sum_{\xi' \in \Xi^q:\xi'(\theta)=\frac{j}{q}} \gamma(\xi')  \right) \mathbb{E}_{\tilde\xi \sim \gamma^{\theta,j}, V \sim \V} \left[ g( V\tilde\xi  ) \right] \\
		& \hspace{-2cm} = \sum_{j: \left| \frac{j}{q}-\xi(\theta) \right| \leq \frac{\epsilon}{4}} \frac{j}{q}  \left( \sum_{\xi' \in \Xi^q:\xi'(\theta)=\frac{j}{q}} \gamma(\xi')  \right) \left[ \left(1-\frac{\alpha}{2} \right) \mathbb{E}_{V \sim \V} \Big[ h(V\tilde\xi) \Big]-\delta\epsilon \right]\quad \textnormal{(By Def.~\ref{def:stability} -- Lem.~\ref{lm:stabilty})}\\
		&\hspace{-2cm} =  \left[ \left(1-\frac{\alpha}{2} \right) \mathbb{E}_{V \sim \V} \Big[ h(V\tilde\xi) \Big]-\delta\epsilon \right] \sum_{j: \left| \frac{j}{q}-\xi(\theta) \right| \leq \frac{\epsilon}{4}} \frac{j}{q}  \sum_{\xi' \in \Xi^q:\xi'(\theta)=\frac{j}{q}} \gamma(\xi')   \\	
		&\hspace{-2cm} =  \left[ \left(1-\frac{\alpha}{2} \right) \mathbb{E}_{V \sim \V} \Big[ h(V\tilde\xi) \Big]-\delta\epsilon \right]  \left(\xi(\theta) -  \sum_{j: \left| \frac{j}{q}-\xi(\theta) \right| \geq \frac{\epsilon}{4}} \sum_{\xi' \in \Xi^q:\xi'(\theta)=\frac{j}{q}} \gamma(\xi')  \right) \\
		& \hspace{-2cm} \geq \xi(\theta) \left[ \left(1-\alpha \right) \mathbb{E}_{V \sim \V} \Big[ h(V\tilde\xi) \Big]-\delta\epsilon \right]. \quad \textnormal{(By Lemma~\ref{lm:small}), $( 1-{\alpha}/{2} )^2 \geq 1-\alpha$, and $\alpha<1$)}
	\end{align*}
	%
	This concludes the proof.
\end{proof}

\section{Proofs omitted from Section~\ref{sec:non_bayesian}}

\lemmaone*

\begin{proof}
	In order to prove the lemma, we show an even stronger result: for every price vector $p \in [0,1]^n$, it is always possible to recover another price vector $p' \in [0,1]^n$ that provides the seller with an expected revenue at least as large as that provided by $p$, and such that $p'_i \geq  \textsc{Rev}_{>i}(\V,p')$ for every $i\in \N$.
	Let us assume that $p$ does \emph{not} satisfy the required condition for some buyer $i \in \N$.
	Then, let $p'$ be such that $p'_i = \textsc{Rev}_{>i}(\V,p)> p_i$ and $p'_j=p_j$ for all $j \in \N: j \neq i$.
	Since by construction $\textsc{Rev}_{>i}(\V,p')=\textsc{Rev}_{>i}(\V,p)$, the condition $p'_i \ge \textsc{Rev}_{>i}(\V,p')$ holds.
	%
	%
	Moreover, the seller's expected revenue for $p'$ in the auction restricted to all buyers $j \in \N: j \geq i$, namely $\textsc{Rev}_{\geq i}(\V,p')$, is such that:
	\begin{align*}
	\textsc{Rev}_{\geq i}(\V,p') & = \text{Pr}_{v_i \sim \V_i} \left\{ v_i \geq p'_i \right\} p'_i + \left( 1-\text{Pr}_{v_i \sim \V_i} \left\{  v_i \geq p'_i \right\} \right) \textsc{Rev}_{>i}(\V,p') \\
	& = \text{Pr}_{v_i \sim \V_i} \left\{ v_i \geq p'_i \right\} p'_i + \left( 1-\text{Pr}_{v_i \sim \V_i} \left\{  v_i \geq p'_i \right\} \right) \textsc{Rev}_{>i}(\V,p) \\
	& = \text{Pr}_{v_i \sim \V_i} \left\{ v_i \geq p_i \right\} p'_i + \left( 1-\text{Pr}_{v_i \sim \V_i} \left\{  v_i \geq p_i \right\} \right) \textsc{Rev}_{>i}(\V,p) \\
	& \geq \text{Pr}_{v_i \sim \V_i} \left\{ v_i \geq p_i \right\} p_i + \left( 1-\text{Pr}_{v_i \sim \V_i} \left\{  v_i \geq p_i \right\} \right) \textsc{Rev}_{>i}(\V,p) \\
	& = \textsc{Rev}_{\geq i}(\V,p), 
	\end{align*}
	where the first equality and the last one holds by definition of $\textsc{Rev}_{\geq i}(\V,p')$, the second one follows from $\textsc{Rev}_{>i}(\V,p')=\textsc{Rev}_{>i}(\V,p)$, the third one holds since $p'_i = \textsc{Rev}_{>i}(\V,p)$, while the inequality follows from $p'_i \geq p_i$.
	As a result, we can conclude that $\textsc{Rev}(\V,p') \geq \textsc{Rev}(\V,p)$.
	The lemma is readily proved by iteratively applying the procedure described above until we get a price vector $p' \in [0,1]^n$ such that $p'_i \geq  \textsc{Rev}_{>i}(\V,p')$ for every buyer $i\in \N$, starting from an optimal price vector $p^*\in \argmax_{p \in [0,1]^n} \textsc{Rev}(\V,p)$.
	%
	%
\end{proof}

\lemmatwo*

\begin{proof}
	Let $p^* \in [0,1]^n$ be a price vector such that $p^* \in \argmax_{p \in [0,1]^n} \textsc{Rev}(\V,p)$ and $p^*_i \geq  \textsc{Rev}_{>i}(\V,p^*)$ for every $i\in \N$.
	Such price vector is guaranteed to exist by Lemma~\ref{lm:bigpayments}.
	We show by induction that $\textsc{Rev}_{\geq i}(\V^\epsilon,p^{*,\epsilon})\geq \textsc{Rev}_{\geq i}(\V,p)-\epsilon$.
	As a base case, it is easy to check that
	\begin{align*}
	\textsc{Rev}_{\geq n}(\V^\epsilon,p^{*,\epsilon}) & =  [p^*_n-\epsilon]_+  \, \text{Pr}_{v_n \sim \V_n^\epsilon} \left\{ v_n \geq [p^*_n-\epsilon]_+  \right\}  \\
	& \geq (p^*_n-\epsilon) \text{Pr}_{v_n \sim \V_n} \left\{ v_n \geq p^*_n \right\} \\
	& \geq p^*_n \, \text{Pr}_{v_n \sim \V_n} \left\{ v_n \geq p^*_n \right\} - \epsilon \\
	& = \textsc{Rev}_{\geq n}( \V,p^*) - \epsilon.
	\end{align*}
	By induction, assume that the condition holds for $i+1$ (notice that $\textsc{Rev}_{>i}(\cdot,\cdot) = \textsc{Rev}_{\geq i+1}(\cdot,\cdot)$), then
	\begin{align*}
	\textsc{Rev}_{\ge i}(\V^\epsilon,p^{*,\epsilon}) & = [p^*_i-\epsilon]_+ \text{Pr}_{v_i \sim \V^\epsilon_i} \left\{ v_i \ge [p^*_i-\epsilon]_+ \right\} + \left( 1-\text{Pr}_{v_i \sim \V^\epsilon_i} \left\{ v_i \ge [p^*_i-\epsilon]_+ \right\} \right) \textsc{Rev}_{>i}(\V^\epsilon,p^{*,\epsilon}) \\
	& \geq (p^*_i-\epsilon) \text{Pr}_{v_i \sim \V^\epsilon_i} \left\{ v_i \ge [p^*_i-\epsilon]_+ \right\} + \left(1-\text{Pr}_{v_i \sim \V^\epsilon_i} \left\{ v_i \ge [p^*_i-\epsilon]_+ \right\} \right) \, \left( \textsc{Rev}_{>i}(\V,p^*)-\epsilon \right)  \\
	& = p^*_i \text{Pr}_{v_i \sim \V^\epsilon_i} \left\{  v_i \ge [p^*_i-\epsilon]_+ \right\} + \left(1-\text{Pr}_{v_i \sim \V^\epsilon_i} \left\{ v_i \ge [p^*_i-\epsilon]_+ \right\} \right)  \textsc{Rev}_{>i}(\V,p^*)-\epsilon \\
	& \geq p^*_i \text{Pr}_{v_i \sim  \V_i} \left\{ v_i \ge p^*_i \right\} + \left(1-\text{Pr}_{v_i \sim  \V_i} \left\{ v_i \ge p^*_i \right\} \right) \, \textsc{Rev}_{>i}(\V,p^*) - \epsilon \\
	&  = \textsc{Rev}_{\ge i}( \V,p^*)-\epsilon,
	\end{align*}
	where the last inequality follows from $p^*_i \geq \textsc{Rev}_{>i}(\V,p^*)$ and $\pr_{v_i \sim \V^\epsilon_i} \left\{ v_i \geq [p^*_i-\epsilon]_+ \right\} \geq \pr_{v_i \sim \V_i} \left\{ v_i \ge p^*_i \right\}$.
\end{proof}

\lemmathree*

\begin{proof}
	Letting $b \coloneqq \lceil \frac{2}{\epsilon}\rceil$ and $K \coloneqq \frac{8}{\epsilon^2} \log \frac{2b^n}{\tau} \in \text{poly}\left( n, \frac{1}{\epsilon}, \log \frac{1}{\tau} \right)$, the proof unfolds in two steps.
	
	The first step is to show that restricting price vectors to those in the discretized set $\mathcal{P}^b$ results in a small reduction of the seller's expected revenue.
	Formally, we prove that:
	\[ 
	\max_{p \in \mathcal{P}^b} \textsc{Rev}(\V,p) \geq \max_{p \in [0,1]^n} \textsc{Rev}(\V,p)-\frac{\epsilon}{2}.
	\]
	To do so, we define some modified distributions of buyers' valuations, namely $ \V^b = \{ \V_i^b \}_{i \in \N}$, which are supported on the discretized set $\mathcal{P}^b$ and are obtained by mapping each valuation $v_i \in [0,1]$ in the support of $\V_i$ (for any $i \in \N$) to a discretized valuation $\frac{x}{b}$, where $x$ is the greatest integer such that $\frac{x}{b}\le  v_i $.
	It is easy to see that, since an optimal price vector for distributions $\V^b$ must specify prices that are multiples of $\frac{1}{b}$, then
	\[
	\max_{p \in [0,1]^n} \textsc{Rev}(\V^b,p) =\max_{p \in \mathcal{P}^b}\textsc{Rev}( \V^b,p).
	\]
	Moreover, by definition of $b$, distributions $\V^b$ are such that $\pr_{v_i \sim  \V^b_i}  \left\{ v_i \ge p_i-\frac{\epsilon}{2} \right\} \geq \pr_{v_i \sim \V_i} \left\{ v_i \geq p_i \right\}$ for every $i \in \N$ and possible price $p_i \in [0,1]$.
	Thus, by Lemma~\ref{lm:smallDecrease}, $\max_{p \in [0,1]^n} \textsc{Rev}( \V^b,p) \ge \max_{p \in [0,1]^n} \textsc{Rev}( \V,p)-\frac{\epsilon}{2}$, which implies that $\max_{p \in \mathcal{P}^b} \textsc{Rev}(\V,p) \ge \max_{p \in [0,1]^n} \textsc{Rev}(\V,p)-\frac{\epsilon}{2}$.
	This proves that we can restrict the attention to price vectors in $\mathcal{P}^b\subset [0,1]^n$, loosing only an additive factor $\frac{\epsilon}{2}$ of the seller's optimal expected revenue.
	
	The second step of the proof is to show that replacing distributions $\V$ with the empirical distributions $\V^K$ built by Algorithm~\ref{alg:computePrices} only reduces the seller's optimal expected revenue by a small amount, with high probability.
	%
	%
	%
	For any price vector $p \in [0,1]^n$, by using Hoeffding's bound we obtain that
	\[
	\pr \left\{ \left| \textsc{Rev}(\V,p)- \textsc{Rev}(\V^K,p) \right| \geq \frac{\epsilon}{4} \right\} \le 2e^{-K \epsilon^2 / 8},
	\]
	where the probability is with respect to the stochasticity of the algorithm (as a result of the sampling steps).
	Since the number of elements in the discretized set $\mathcal{P}^b$ is $b^n$, by a union bound we get
	\[
	\pr \left\{  \left| \textsc{Rev}(\V,p)- \textsc{Rev}(\V^K,p) \right| <  \frac{\epsilon}{4} \quad \forall p \in \mathcal{P}^b \right\} \ge 1- 2 b^{n} e^{-K\epsilon^2/8}=1-\tau.
	\]
	Letting $p \in \mathcal{P}^b$ be the price vector returned by Algorithm~\ref{alg:computePrices}, it is the case that $p \in \argmax_{p' \in \mathcal{P}^b} \textsc{Rev}(\V^K,p')$, given the correctness and optimality of the backward induction procedure with which the vector $p$ is built~\citep{TaoComplexity2020}.
	Moreover, letting $p^* \in \argmax_{p' \in \mathcal{P}^b}\textsc{Rev}(\V,p')$ be an optimal price vector over the discretized set $\mathcal{P}^b$ for the actual distributions of buyers' valuations $\V$, with probability at least $1-\tau$ it holds that
	\[
	\textsc{Rev}(\V,p)\ge  \textsc{Rev}(\V^K,p)-\frac{\epsilon}{4} \ge \textsc{Rev}(\V^K,p^*) -\frac{\epsilon}{4} \ge \textsc{Rev}(\V,p^*) -\frac{\epsilon}{2}.
	\]
	Hence, with probability at least $1- \tau$, it holds $\textsc{Rev}(\V,p)\ge \textsc{Rev}(\V,p^*) -\frac{\epsilon}{2}\ge \max_{p' \in [0,1]^n}\textsc{Rev}(\V,p') -\epsilon$, where the last step has been proved in the first part of the proof.
	
	In order to conclude the proof, it is sufficient to notice that, with probability at least $1-\tau$, it also holds that 
	$r  = \textsc{Rev}(\V^K,p) \in \left[ \textsc{Rev}(\V,p)-\frac{\epsilon}{4} ,\textsc{Rev}(\V,p)+ \frac{\epsilon}{4}  \right]$.
\end{proof}

\section{Proofs omitted from Section~\ref{sec:public}}

\lemmafour*

\begin{proof}
	As a first step, we prove the following: given any matrix $V\in [0,1]^{n\times d}$ of buyers' valuations and any price vector $p' \in [0,1]^n$, for every distribution $\gamma$ over $\Delta_{\Theta}$ that is $(\alpha,\epsilon)$-decreasing around $\xi$ (see Definition~\ref{def:perturbation}) it holds that
	\begin{equation}\label{eq:boundRevenue}
	\mathbb{E}_{\tilde\xi \sim  \gamma} \left[ g_{p'}(V \tilde\xi) \right] \ge \mathbb{E}_{\tilde \xi\sim \gamma} \left[ g_{p'} \left(\max \left\{V \tilde \xi, V \xi -\epsilon \mathbbm{1} \right\} \right) \right]- \alpha \, g_{p'+\epsilon \mathbbm{1}} (V \xi).
	\end{equation}
	W.l.o.g., let $i \in \N$ be the buyer that buys the item when buyers' valuations are specified by the vector $V \xi - \epsilon \mathbbm{1}$ and the proposed prices are those specified by $p'$, that is, it must be the case that $p'_i \leq V_i \xi -\epsilon$ and $p'_j > V_j \xi - \epsilon$ for all $j \in \N: j<i$.
	Since $\gamma$ is $(\alpha,\epsilon)$-decreasing around $\xi$, by sampling a posterior $\tilde\xi \in \Delta_{\Theta}$ according to $\gamma$, with probability at least $1-\alpha$ it holds that $V_{i} \tilde\xi \ge V_{i}\xi-\epsilon$ (see Definition~\ref{def:perturbation}).
	Moreover, let $\tilde\Xi \coloneqq \{ \tilde\xi \in \Delta_{\Theta} \mid V_{i}\tilde\xi \geq V_{i}\xi - \epsilon  \}$ be the set of posteriors which result in a buyer $i$'s valuation that is at most $\epsilon$ less than that for $\xi$ (notice that $\sum_{\tilde\xi \in \tilde\Xi} \gamma(\tilde\xi) \geq 1-\alpha$).
	Then, we split the posteriors in $\Delta_{\Theta}$ into three groups, as follows:
	\begin{itemize}
		\item $\Xi^1 \subseteq \tilde\Xi$ is composed of all the posteriors $\tilde\xi \in \tilde\Xi$ such that, for every $j \in \N: j <i$, it holds $V_j \tilde\xi < p'_j$;
		\item $\Xi^2 \subseteq \tilde\Xi$ is composed of all the posteriors $\tilde\xi \notin \tilde\Xi$ such that, for every $j \in \N: j <i$, it holds $V_j \tilde\xi < p'_j$;
		\item $\Xi^3 \subseteq \Delta_{\Theta}$ is composed of all the posteriors $\tilde\xi \in \Delta_{\Theta}$ for which there exists a buyer $j(\tilde\xi) \in \N : j < i$ (notice the dependence on $\tilde\xi$) such that $j(\tilde\xi) = \min \{ j \in \N \mid  V_j \tilde\xi \ge p'_j \}$
	\end{itemize}
	Next, we show that, for every posterior $\tilde\xi \in \Xi^1 \cup \Xi^3$, it holds $g_{p}(V \tilde\xi) = g_{p'} (\max\{V\tilde\xi,V \xi -\epsilon \mathbbm{1} \})$, while, for every $\tilde\xi \in \Xi^2$, it holds $g_{p'} (\max\{ V\tilde\xi,V\xi -\epsilon \mathbbm{1} \}) \le g_{p' + \epsilon \mathbbm{1} }(V\xi)$.
	First, let us consider a posterior $\tilde\xi \in \Xi^1$.
	For each $j \in \N : j<i$, it holds $V_j \tilde\xi \leq \max \{V_j \tilde\xi, V_j \xi-\epsilon\} <p'_j$ (by definition of $\Xi^1$, and since buyer $j$ does not buy the item for price $p'_j$).
	Moreover, since $V_{i}\tilde\xi \geq V_{i} \xi-\epsilon $, it holds $V_{i}\tilde\xi=\max\{V_{i}\tilde\xi,V_{i}\xi-\epsilon\}\ge p'_{i}$.
	Hence, both when buyers' valuations are specified by the vector $V \tilde\xi$ and when they are given by $\max\{V \tilde\xi, V \xi-\epsilon \mathbbm{1} \}$ (with $\max$ applied component-wise), it is the case that buyer $i$ buys the item at price $p'_{i}$, resulting in
	\[
	g_{p'}(V \tilde\xi) = g_{p'}(\max\{V \tilde\xi, V \xi-\epsilon \mathbbm{1}\}).
	\]
	
	Now, let us consider a posterior $\tilde\xi \in \Xi^2$.
	In this case, $\max \{ V_{i} \tilde\xi, V_{i} \xi -\epsilon \}=V_{i} \xi - \epsilon \ge p'_{i}$, while $ \max \{ V_{j} \tilde\xi,V_{j}\xi-\epsilon  \}<p'_j$ for every $j \in \N: j<i$.
	Thus, both when buyers' valuations are specified by $\max \{ V\tilde\xi,V\xi-\epsilon\mathbbm{1} \}$ and when they are given by $V\xi-\epsilon\mathbbm{1}$, it is the case that buyer $i$ buys the item at price $p'_i$, resulting in
	\[
	g_{p'}(\max\{V \tilde\xi,V\xi-\mathbbm{1}\})=g_{p'}(V\xi-\epsilon \mathbbm{1})\le g_{p'+\epsilon \mathbbm{1}}(V\xi),
	\] 
	where the inequality holds since buyer $i$ buys the item at price $p'_{i}$ for valuations $V\xi-\epsilon \mathbbm{1}$ and price vector $p'$, while the buyer would still buy the item, though at price $p'_{i}+\epsilon \geq p'_i$, for valuations $V\xi$ and price vector $p' + \epsilon \mathbbm{1}$.
	Finally, let us consider a posterior $\tilde\xi \in \Xi^3$.
	We have that, for every $j \in \N : j<j(\tilde\xi)$, it holds $V_j \tilde\xi \le \max \{V_j\tilde\xi, V_{j} \xi-\epsilon \} < p'_j$, while $\max\{V_{j(\tilde\xi)}\tilde\xi, V_{j(\tilde\xi)}\xi-\epsilon \}  \ge V_{j(\tilde\xi)}\tilde\xi\ge p'_{j(\tilde\xi)}$.
	As a result, both when buyers' valuations are specified by $V \tilde\xi$ and when they are given by $\max\{ V \tilde\xi,V\xi-\epsilon\mathbbm{1} \}$, it is the case that buyer $j(\tilde\xi)$ buys the item at price $p'_{j(\tilde\xi)}$, resulting in
	\[
	g_{p'}(V \tilde\xi) =  g_{p'} (\max\{V\tilde\xi,V\xi-\epsilon\mathbbm{1} \}).
	\]
	%
	%
	%
	This allows us to prove Equation~\eqref{eq:boundRevenue}, as follows:
	\[
	\mathbb{E}_{\tilde\xi \sim \gamma} \Big[ g_{p'}(\max\{V \tilde\xi,V\xi-\epsilon \mathbbm{1} \}) \Big]-\mathbb{E}_{\tilde \xi\sim  \gamma} \Big[ g_{p'}(V \tilde\xi) \Big]\leq
	\sum_{\tilde\xi \in \Xi^2} \gamma(\tilde\xi) g_{p'+\epsilon\mathbbm{1} }(V\xi) \leq
	\alpha g_{p'+\epsilon\mathbbm{1} }(V\xi),
	\]
	where the first inequality comes from the fact that, as previously proved, $g_{p'}(\max\{V \tilde\xi,V\xi-\epsilon\mathbbm{1}\})=g_{p'}(V \tilde\xi)$ for every posterior $\tilde\xi \in \Xi^1 \cup \Xi^3$ and $g_{p'}(\max\{V  \tilde\xi,V\xi-\epsilon \mathbbm{1}\})\leq g_{p' + \epsilon\mathbbm{1}}(V\xi)$ for every posterior $\tilde\xi \in \Xi^2$, while the second inequality is readily obtained by noticing that $\sum_{\tilde\xi \in  \Xi^2} \gamma(\tilde\xi) \leq \sum_{\tilde\xi \notin \tilde\Xi} \gamma(\tilde\xi) \leq \alpha$.
	
	Given any posterior $\tilde\xi \in \Delta_{\Theta}$, the expression $\max_{p' \in [0,1]^n} \mathbb{E}_{V \sim \V} \Big[ \textsc{Rev}(\max \{ V \tilde\xi,V\xi-\epsilon\mathbbm{1} \},p') \Big]$ can be interpreted as the optimal seller's expected revenue when buyers' valuations are determined by distributions $\V^\epsilon = \{ \V_i^\epsilon \}$ such that, for every buyer $i \in \N$, their valuation is sampled by first drawing a valuation $v_i \in [0,1]^d$ according to $\V_i$ and, then, taking $\max\{ v_i^\top\tilde\xi, v_i^\top \xi -\epsilon \}$.
	Moreover, $\max_{p' \in [0,1]^n} \mathbb{E}_{V \sim \V}  \textsc{Rev}(V \xi, p')$ can be interpreted as the optimal seller's expected revenue when buyers' valuations are determined by distributions $\V = \{ \V_i \}$ such that valuations are determined by first sampling a $v_i \in [0,1]^d$ from $\V_i$ and, then, taking $v_i ^\top \xi$.
	It is easy to see that $\text{Pr}_{v_i \sim \V_i^\epsilon} \left\{ v_i\ge p_i'-\epsilon \right\}\ge \text{Pr}_{v_i\sim \V_i} \left\{ v_i\ge p_i' \right\}$ for every price $p_i'$, so that distributions $\V^\epsilon$ and $\V$ satisfy Definition~\ref{def:distribs}.
	%
	%
	Then, by applying Lemma~\ref{lm:smallDecrease}, we can conclude that there exists a price vector $p \in [0,1]^n$ such that $\textsc{Rev}(\V^\epsilon,p) \geq \max_{p' \in [0,1]^n} \textsc{Rev}(\V,p')-\epsilon$.
	%
	%
	Thus, for every distribution $\gamma$ over $\Delta_{\Theta}$ that is $(\alpha,\epsilon)$-decreasing around $\xi$, we get
	\begin{align*}
	\mathbb{E}_{\tilde\xi \sim \gamma,  V \sim \V} \Big[  g_{p}(V \tilde\xi ) \Big] & \ge \mathbb{E}_{\tilde\xi \sim \gamma,  V \sim \V} \Big[  g_{p}( \max\{ V \tilde\xi,V\xi-\epsilon\mathbbm{1} \} ) \Big] - \mathbb{E}_{ V \sim \V} \Big[ \alpha \, g_{p +\epsilon\mathbbm{1} }( V\xi) \Big]  \\
	& \ge \max_{p' \in [0,1]^n} \mathbb{E}_{V \sim \V} \Big[ \textsc{Rev}(  V\xi,p') \Big] -\epsilon-  \mathbb{E}_{V \sim \V} \Big[ \alpha \, g_{p +\epsilon\mathbbm{1} }( V\xi) \Big] \\
	& \ge \max_{p' \in [0,1]^n} \mathbb{E}_{ V \sim \V } \Big[  \textsc{Rev} ( V \xi, p') \Big]-\epsilon- \max_{p' \in [0,1]^n} \mathbb{E}_{V \sim \V} \Big[  \alpha \, \textsc{Rev}( V \xi, p') \Big] \\
	& \ge (1-\alpha) \max_{p' \in [0,1]^n} \mathbb{E}_{V \sim \V} \Big[ \textsc{Rev}(V\xi, p') \Big]-\epsilon,
	\end{align*}
	where the first inequality holds by Equation~\eqref{eq:boundRevenue}, while the second one by Lemma~\ref{lm:smallDecrease}.
\end{proof}

\lemmafive*

\begin{proof}
	The idea of the proof is to build a signaling scheme $\phi$ such that there is one-to-one correspondence between the buyers' posteriors induced by signal profiles $s \in \sset$ under $\phi$ and the posteriors in the support of the distribution $\gamma$.
	Thus, in the following we can safely use the notation $\xi_s$ to the denote the posterior corresponding to signal profile $s \in \sset$.
	We define the signaling scheme $\phi : \Theta \to \Delta_{\sset}$ so that, for every state $\theta \in \Theta$, it holds $\phi_\theta(s)=\frac{\gamma(\xi_s) \xi_s(\theta)}{\mu_\theta}$ for all $s \in \sset$.
	Moreover, we define $f: \sset \to [0,1]^n$ so that $f(s) = f^\circ(\xi_s)$ for all $s\in \sset$.
	First, notice that the signaling scheme $\phi$ is consistent, since, for every $\theta \in \Theta$, it holds $\sum_{s \in \sset} \phi_\theta(s) =\sum_{s \in \sset} \frac{\gamma(\xi_s) \xi_s(\theta) }{\mu_\theta}= \sum_{\xi \in \text{supp}(\gamma)} \frac{\gamma(\xi) \xi(\theta) }{\mu_\theta}=1$, where the last two equalities follow from the correspondence between signal profiles and posteriors in $\text{supp}(\gamma)$ and the fact that $\sum_{\xi \in \text{supp}(\gamma)} \gamma(\xi) \xi(\theta)= \mu_\theta$.
	It is also easy to check that each signal profile $s \in \sset$ indeed induces its corresponding posterior $\xi_s$ under the signaling scheme $\phi$.
	Finally, we have
	\[
	\sum_{\theta \in \Theta} \mu_\theta \sum_{s \in \sset}  \phi_\theta(s) \textsc{Rev}(\V,f(s),\xi_s) = \sum_{\theta \in \Theta} \mu_\theta \sum_{\xi \in \text{supp}(\gamma)} \frac{\gamma(\xi) \xi(\theta) }{\mu_\theta} \textsc{Rev}(\V,f^\circ(\xi),\xi) =\sum_{\xi \in \text{supp}(\gamma)} \gamma(\xi) \textsc{Rev}(\V,f^\circ(\xi),\xi),
	\]
	which concludes the proof.
	%
	%
\end{proof}

\lemmasix*

\begin{proof}
	Given a Bayesian posted price auction with prior $\mu \in \Delta_{\Theta}$ and distributions of buyers' valuations $\V$, let $(\phi^*,f^*)$ be a revenue-maximizing signaling scheme, price function pair.
	Then, we define $\gamma^*$ as the probability distribution over posteriors $\Delta_{\Theta}$ induced by $\phi^*$. 
	Moreover, we define $f^{\circ,*} : \text{supp}(\gamma^*) \to [0,1]^n$ in such a way that, for every posterior $\xi \in \text{supp}(\gamma^*)$, it holds $f^{\circ,*}(\xi) = f^*(s)$, where $s \in \sset$ is the signal inducing $\xi$, namely $\xi = \xi_s$.\footnote{W.l.o.g., we can safely assume that there is a unique signal inducing $\xi$. Indeed, if two signals $s \in \sset$ and $s' \in \sset$ induce the same posterior, then it is possible to build another signaling scheme, price function pair $(\phi^*,f^*)$ that joins the two signals in a new single signal $s^* \in \sset$, by setting $\phi^*_\theta(s^*) \gets \phi^*_\theta(s)+\phi^*_\theta(s')$ and $f^*(s^*)=f(s)$ if $\textsc{Rev}(\V,f^*(s),\xi)\geq \textsc{Rev}(\V,f^*(s'),\xi)$, while $f^*(s)=f^*(s')$ otherwise. It is easy to check that the new signaling scheme cannot decrease the seller's expected revenue.}
	%

	Let $\alpha=\epsilon=\frac{\eta}{2}$ and $q=\frac{32\log \frac{4}{\alpha}}{\epsilon^2}$.
	Then, we build a probability distribution $\gamma$ over posteriors in $\Delta_{\Theta}$ by decomposing each posterior $\xi \in \text{supp}(\gamma^*)$ according to Corollary~\ref{lm:decompositionweak}.
	Additionally, each time we decompose a posterior, for every newly-introduced posterior $\xi \in \Delta_{\Theta}$ we define the function $f^\circ : \Delta_{\Theta} \to [0,1]^n$ so that $f^\circ(\xi) \in \argmax_{p \in [0,1]^n} \textsc{Rev}(\V,p,\xi)$.
	Letting $\gamma^{\xi} \in \Delta_{\Xi^q}$ be the probability distribution over $q$-uniform posteriors which is obtained by decomposing posterior $\xi \in \text{supp}(\gamma^*)$ according to Corollary~\ref{lm:decompositionweak}, we define $\gamma$ so that $\gamma(\xi) = \sum_{\xi' \in \text{supp}(\gamma^*)} \gamma^*(\xi') \gamma^{\xi'}(\xi)$ for every $\xi \in \Xi^q$.
	%
	
	First, let us notice that, for every $\theta \in \Theta$, it holds
	\[
	\sum_{\xi \in \Xi^q} \gamma(\xi) \xi(\theta)= \sum_{\xi' \in \text{supp}(\gamma^*)} \gamma^*(\xi') \sum_{\xi \in \Xi^q}  \gamma^{\xi'}(\xi) \xi(\theta) = \sum_{\xi' \in \text{supp}(\gamma^*)} \gamma^*(\xi') \xi'(\theta)= \mu_\theta,
	\]
	where the second equality follows from the property of the decomposition in Theorem~\ref{thm:decomposition}, while the last one from the fact that $\gamma^*$ is induced by a signaling scheme.
	%
	%
	Moreover, given any posterior $\xi \in \text{supp}(\gamma^*)$, let $p^\xi \in [0,1]^n$ be a price vector such that, for every $p \in [0,1]^n$, the function $g_{p^\xi}$ is $(1,\alpha,\epsilon)$-stable compared with the function $g_p$ in $(\V,\xi)$.
	such price vectors are guaranteed to exist by Lemma~\ref{lm:stabAuction}.
	Then, the pair $(\gamma, f^\circ)$ provides the seller with an expected revenue of 
	\begin{align*}
	\sum_{\xi \in \Xi^q} \gamma(\xi) \textsc{Rev}(\V,f^\circ(\xi), \xi) & =\sum_{\xi \in \text{supp}(\gamma^*)} \gamma^*(\xi) \sum_{\xi' \in \Xi^q} \gamma^{\xi}(\xi') \textsc{Rev}(\V,f^\circ(\xi),\xi)  \\
	& = \sum_{\xi \in \text{supp}(\gamma^*)} \gamma^*(\xi) \sum_{\xi' \in \Xi^q} \gamma^{\xi}(\xi') \max_{p \in [0,1]^n} \textsc{Rev}(\V,p,\xi) \\
	& \ge \sum_{\xi \in \text{supp}(\gamma^*)} \gamma^*(\xi)  \sum_{\xi' \in \Xi^q}  \gamma^{\xi}(\xi') \textsc{Rev}(\V,p^\xi,\xi)  \\
	&  \ge \sum_{\xi \in \text{supp}(\gamma^*)} \gamma^*(\xi)  \left[ (1-\alpha)    \textsc{Rev}(\V, f^{\circ,*}(\xi), \xi)- \epsilon \right] \\
	& = (1-\alpha) \sum_{\xi \in \text{supp}(\gamma^*)} \gamma^*(\xi) \textsc{Rev}(\V, f^{\circ,*}(\xi),\xi)- \epsilon \\
	& = \left( 1-\frac{\eta}{2} \right) \sum_{\xi \in \text{supp}(\gamma^*)} \gamma^*(\xi) \textsc{Rev}(\V, f^{\circ,*}(\xi),\xi)- \frac{\eta}{2} \\
	& \ge \sum_{\xi \in \text{supp}(\gamma^*)} \gamma^*(\xi) \textsc{Rev}(\V,f^{\circ,*}(\xi),\xi) - \eta \\
	&  = \sum_{\theta \in \Theta} \mu_\theta \sum_{s \in \sset} \phi^*_\theta(s)  \textsc{Rev}(\V, f^*(s), \xi_s) -   \eta,
	\end{align*}
	%
	which allows us to conclude that there exists a pair $(\phi,f)$ that only uses $q$-uniform posteriors and provides the seller with an expected revenue arbitrary close to that of an optimal pair.
\end{proof}

\theoremthree*

\begin{proof}
	By Lemma~\ref{lm:optQUniform}, given any constant $\eta > 0$ and letting $q = \frac{128\log \frac{6}{\eta}}{\eta^2}$, LP~\ref{eq:lppublic} has optimal value at least $OPT-\eta$.
	%
	The polynomial-time algorithm that proves the theorem solves an approximated version of LP~\ref{eq:lppublic}, which is obtained by replacing the terms $\max_{p \in [0,1]^n} \textsc{Rev}(\V,p,\xi)$ with suitable values $U(\xi)$.
	The latter are obtained by running Algorithm~\ref{alg:computePrices} (the values of $\epsilon$ and $\tau$ are defined in the following) for the (non-Bayesian) auctions in which the buyers' valuations are those resulting by multiplying samples drawn from distributions $\V_i$ by the posterior $\xi$.
	We let $(p^\xi,U(\xi))$ be the pair returned by Algorithm~\ref{alg:computePrices}.
	%
	By Lemma~\ref{lm:samplingPublic}, for every $q$-uniform posterior $\xi \in \Xi^q$, Algorithm~\ref{alg:computePrices} runs in polynomial time and the price vector $p^\xi$ is such that, with probability at least $1-\tau$, it holds
	\[
	\mathbb{E}_{V \sim \V} \Big[ g_{p^\xi}( V \xi) \Big] \ge \max_{p \in [0,1]^n} \mathbb{E}_{ V \sim \V} \Big[ g_{p}( V\xi) \Big]-\epsilon \,\, \text{and} \,\, U(\xi)\in \left[ \mathbb{E}_{ V \sim \V} \Big[ g_{p^\xi}( V\xi) \Big] - \epsilon, \mathbb{E}_{ V \sim \V} \Big[ g_{p^\xi}( V\xi) \Big] + \epsilon \right].
	\]
	As a result, with probability at least $1-\tau|\Xi^q|$, the previous conditions hold for every $q$-uniform posterior.

	Next, we show that, with probability at least $1-\tau|\Xi^q|$, an optimal solution to LP~\ref{eq:lppublic} is close to an optimal solution of the following LP obtained by replacing the $\max$ terms in the objective of LP~\ref{eq:lppublic} with the values $U(\xi)$:
	\begin{subequations}\label{eq:lp2}
		\begin{align}
		\max_{\gamma \in \Delta_{\Xi^q}} & \sum_{\xi \in \Xi^q}   \gamma(\xi)U(\xi) \quad \textnormal{s.t.} \\
		&\sum_{\xi \in \Xi^q} \gamma(\xi) \, \xi(\theta) =\mu_\theta & \forall \theta \in \Theta.
		\end{align}
	\end{subequations}
	Notice that, for a constant $q \in \mathbb{N}_{>0}$, the number of $q$-uniform posteriors is at most $d^q$, so that LP~\ref{eq:lp2} can be solved in polynomial time, as it involves ${O}(|\Xi^q|)$ variables and constraints.
	
	Let $(\gamma,f^\circ )$ be such that $\gamma \in \Delta_{\Xi^q}$ is an optimal solution to LP~\ref{eq:lp2} and $f^\circ  : \Delta_{\Theta} \to [0,1]^n$ is such that, for every $\xi \in \Xi^q$, it holds $f^\circ (\xi) = p^\xi$, which is the price vector obtained by running Algorithm~\ref{alg:computePrices}.
	Moreover, let $(\gamma^*,f^{\circ, *})$ be an optimal solution to LP~\ref{eq:lppublic}.
	Then, with probability at least $1-\tau|\Xi^q|$,
	\[
	\sum_{\xi \in \Xi^q}   \gamma(\xi) \mathbb{E}_{ V \sim \V} \Big[ g_{f^\circ (\xi)}( V\xi) \Big] \ge \sum_{\xi \in \Xi^q}   \gamma(\xi) U(\xi)-\epsilon \ge \sum_{\xi \in \Xi^q}   \gamma^*(\xi) U(\xi) - \epsilon \ge \sum_{\xi \in \Xi^q}  \gamma^*(\xi) \mathbb{E}_{ V \sim \V} \Big[ g_{f^{\circ ,*}(\xi)}( V\xi,f^{\circ,*}(\xi)) \Big]-2\epsilon.
	\]
	In conclusion, since $\sum_{\xi \in \Xi^q}\gamma^*(\xi) \mathbb{E}_{V \sim \V} \mleft[g_{f^{\circ,*}(\xi)}(V\xi,f^{\circ,*}(\xi))\mright]\ge OPT-\eta$ by Lemma~\ref{lm:optQUniform}, we have:
	\[
	\sum_{\xi \in \Xi^q}   \gamma(\xi) \mathbb{E}_{V \sim \V} \Big[ g_p(V\xi) \Big] \ge OPT-2\epsilon-\eta
	\]
	with probability at least $1-\tau|\Xi^q|$.
	Hence,
	\[
	\mathbb{E} \left[ \sum_{\xi \in \Xi^q}   \gamma(\xi) \textsc{Rev}(\V,f^\circ(\xi),\xi) \right] \ge \left( 1-\tau d^{q} \right) OPT-2\epsilon-\eta,
	\]
	where the expectation is over the randomness of the algorithm.
	Finally, Lemma~\ref{lm:posteriorToSignaling} allows us to recover from $(\gamma,f^\circ)$ a signaling scheme with the same seller's expected revenue.
	For any additive approximation factor $\lambda>0$, setting $\epsilon=\frac{\lambda}{6}$, $\eta=\frac{\lambda}{3}$, and $\tau=\frac{\lambda}{3d^q}$, we obtain the desired approximation bound.
	Moreover, the algorithm runs in polynomial time since $\eta$ is constant and the running time of the algorithm is polynomial in $\epsilon,\tau$ and the size of the problem instance.
\end{proof}

\section{Proofs omitted from Section~\ref{sec:private}}

\lemmaseven*

\begin{proof}
	Let us recall that $g_{i,{p_i}}: [0,1]^n \to \{ 0,1\}$ is such that $g_{i,{p_i}}(x) = \mathbb{I}\{ x_i \geq p_i \}$ for any value $x \in [0,1]^n$.
	As a first step, we show that, for every valuation vector $v_i \in [0,1]^{d}$ and probability distribution $\gamma$ over $\Delta_{\Theta}$ that is $(\alpha,\epsilon)$-decreasing around $\xi$,
	it holds $\mathbb{E}_{\tilde \xi_i \sim \gamma}\mathbb{I}\{ v_i^\top \tilde \xi_i \ge [p_i-\epsilon]_+ \} \ge (1-\alpha) \mathbb{I}\{ v_i^\top \xi_i \ge p_i \}$.
	Two cases are possible.
	If $\mathbb{I}\{ v^\top_i \xi_i \ge p_i \}=0$, then the inequality trivially holds.
	If $\mathbb{I}\{ v^\top_i \xi_i \ge p_i \}=1$, by Definition~\ref{def:s_distribution} we have that, with probability at least $1-\alpha$, a posterior $\tilde\xi_i \in \Delta_{\Theta}$ randomly drawn according to $\gamma$ satisfies $v^\top_i \tilde \xi_i \ge [v^\top_i \xi_i -\epsilon]_+ \ge [p_i-\epsilon]_+$, which implies that $\mathbb{I}\{ v^\top_i\tilde \xi_i \ge [p_i-\epsilon]_+ \}=1$.
	Hence, $\mathbb{E}_{\tilde \xi_i \sim \gamma}\mathbb{I}\{ v^\top_i\tilde \xi_i \ge [p_i-\epsilon]_+ \}\ge (1-\alpha)=(1-\alpha)\mathbb{I}\{ v^\top_i \xi_i \ge p_i \}$, as desired.
	Since $\mathbb{E}_{\tilde \xi_i \sim \gamma}\mathbb{I}\{  v^\top_i \tilde \xi_i \ge [p_i-\epsilon]_+ \} \ge (1-\alpha) \mathbb{I} \{  v^\top_i \xi_i \ge p_i \}$  for every $v_i \in [0,1]^{d}$, by taking the expectation over vectors $v \sim \V$ we obtain
	$\mathbb{E}_{ V \sim \V} \mathbb{E}_{\tilde \xi_i \sim \gamma}\mathbb{I}\{  V_i\tilde \xi_i \ge [p_i-\epsilon]_+ \}\ge (1-\alpha) \mathbb{E}_{ V \sim \V} \mathbb{I}\{ V_i \xi_i \ge p_i \}$, which fulfills the condition in Definition~\ref{def:stability} and proves the result. 
\end{proof}

\lemmaeight*

\begin{proof}
	We define the set of signals for buyer $i \in \N$ as $\sset_i \coloneqq \Xi^q_i \times P^b$.
	Then, we set $\phi : \Theta \to \Delta_{\sset}$ so that, for every $\theta \in \Theta$ and $s \in \sset$, it holds $\phi_\theta(s)=\frac{y_{\theta,\xi,p}}{\mu_\theta}$, where the pair $(\xi,p)$ with $\xi = (\xi_1,\ldots,\xi_n) \in \Xi^q$ and $p \in \mathcal{P}^b$ is such that $(\xi_i,p_i)=s_i$ for each $i \in \N$.
	Moreover, we set $f_i(s_i)=p_i$ for every buyer $i \in \N$ and signal $s_i = (\xi_i,p_i) \in \sset_i$.
	First, we show that $\phi$ is well defined, that is, for every state of nature $\theta \in \Theta$, it holds
	\[
		\sum_{s \in \sset} \phi_\theta(s)=\sum_{\xi \in \Xi^q} \sum_{p \in \mathcal{P}^b} \frac{y_{\theta,\xi,p}}{\mu_\theta}=\sum_{\xi_1 \in \sset_1} \sum_{p_1 \in P^b} \frac{\xi_1(\theta) t_{1,\xi_1,p_1}}{\mu_\theta}=\sum_{\xi_1 \in \Xi^q_1} \frac{\xi_1(\theta) \gamma_{1,\xi_1}}{\mu_\theta}=\frac{\mu_\theta}{\mu_\theta}=1,
	\]
	where we use Constraints~\eqref{eq:private1} in the second equality, Constraints~\eqref{eq:private2} in the third one, and Constraints~\eqref{eq:private3} in the last one.
	Next, we show that, for any $\xi\in \Xi^q$ and $p \in \mathcal{P}^b$, it holds $\textsc{Rev}(\V,f(s),\xi_s)=\textsc{Rev}(\V,p,\xi)$, where the signal profile $s \in \sset$ is such that $s_i = (\xi_i,p_i)$ for every $i \in \N$.
	Clearly, the prices coincide, namely $f(s) = p$.
	Thus, it is sufficient to prove that each signal $s_i = (\xi_i,p_i)$ induces posterior $\xi_i$ for buyer $i \in \N$.
	%
	For every $\theta \in \Theta$, it holds
	\begin{align*} 
		\mu_{\theta } \sum_{s' \in \sset:s'_i=s_i}\phi_\theta(s') =\sum_{\xi' \in \Xi^q,p' \in \mathcal{P}^b:(\xi'_i,p'_i)=s_i}y_{\theta,\xi',p'}=\xi_i(\theta) t_{i,\xi_i,p_i}.
	\end{align*}
	Hence, for every $\theta \in \Theta$, 
	\begin{align*}
		\xi_{i,s_i}(\theta) = \frac{{\mu}_\theta {\phi}_{i, \theta}(s_i) }{\sum_{\theta'\in\Theta}{\mu}_{\theta'} {\phi}_{i, \theta'}(s_i) } = \frac{\mu_\theta \sum_{s' \in \sset:s'_i=s_i} \phi_\theta(s')}{\sum_{\theta' \in \Theta} \mu_{\theta'}\sum_{s' \in \sset:s'_i=s_i} \phi_{\theta'}(s')}=\frac{\xi_i(\theta) t_{i,\xi_i,p_i}}{t_{i,\xi_i,p_i}}=\xi_i(\theta).
	\end{align*}
	%
	Thus, the seller's expected revenue for the pair $(\phi,f)$ is
	\begin{align*}
		\sum_{s \in \sset} \sum_{\theta \in \Theta} \mu_\theta \phi_\theta(s) \textsc{Rev}(\V,f(s),\xi_s)=
		\sum_{\xi \in \Xi^q} \sum_{p \in \mathcal{P}^b} \sum_{\theta \in \Theta} \mu_\theta \frac{y_{\theta,\xi,p}}{\mu_\theta} \textsc{Rev}(\V,p,\xi)=
		\sum_{\theta \in \Theta} \sum_{\xi \in \Xi^q} \sum_{p \in \mathcal{P}^b}  y_{\theta,\xi,p} \textsc{Rev}(\V,p,\xi),
	\end{align*}
	which proves the lemma.
\end{proof}

\lemmanine*

\begin{proof}
	We show that, given a revenue-maximizing pair $(\phi,f)$ (with seller's revenue $OPT$), we can recover an optimal solution to LP~\ref{lp:private} whose value is at least  $OPT-\eta$ when the LP is instantiated with suitable constants $b(\eta) \in \mathbb{N}_{>0}$ and $q(\eta) \in \mathbb{N}_{>0}$ (depending on the approximation level $\eta$).
	%
	%
	%
	Let $\alpha=\epsilon=\frac{\eta}{3}$, $b = \lceil \frac{3}{\eta}\rceil$, and $q=\frac{32 \log \frac{4}{\alpha}}{\epsilon^2}$.
	Recalling that $\xi_{i,s_i} \in \Delta_{\Theta}$ denotes buyer $i$'s posterior induced by signal $s_i \in \sset_i$, we let $\gamma^{s_i} \in \Delta_{\Xi_i^q}$ be the probability distribution over $q$-uniform posteriors obtained by decomposing $\xi_{i,s_i}$ according to Theorem~\ref{thm:decomposition}.
	By Lemma~\ref{lm:singleBuyer} and Theorem~\ref{thm:decomposition}, it follows that, for every $p_i\in [0,1]$ and $\theta \in \Theta$,
	\begin{align}\label{eq:probBound}
		\sum_{\xi_i \in \Xi^q_i} \gamma^{s_i} (\xi_i) \xi_i(\theta) \pr_{ v_i \sim \V_i} \left\{   v_i^\top \xi_i  \ge [p_i-\epsilon]_+ \right\} \ge \xi_{i,s_i}(\theta) (1-\alpha)  \pr_{ v_i \sim \V_i} \left\{  v_i^\top \xi_i\ge p_i \right\}.
	\end{align}
	%

	For every signal profile $s \in \sset$, we define a non-Bayesian posted price auction in which the distributions of buyers' valuations are $\V^s = \{ \V_i^s \}_{i \in \N}$, where each $\V^s_i$ is such that a valuation $v_i \sim \V^s_i$ is obtained by first sampling $\tilde v_i \sim \V_i$ and then letting $v_i=\tilde v_i^{\top} \xi_{i,s_i}$.
	Moreover, we let $p^s \in [0,1]^n$ be a price vector for the seller in such non-Bayesian auction, with $ p^s_i\ge \textsc{Rev}_{> i}(\V^s,p^s)$ for every $i \in \N$.
	By Lemma~\ref{lm:bigpayments}, such a vector always exists.
	Finally, given $p^s$, we let $ \hat p^s \in [0,1]^n$ be such that each price $ \hat p^s_i$ is the greatest price $p_i \in P^b$ (among discretized prices) satisfying the inequality $p_i \le [ p^s_i-\epsilon]_+$; formally,
	\[
		p_i^s = \max \left\{ p_i\in P^b \mid p_i \le [ p^s_i-\epsilon]_+ \right\}.
	\]
	
	Next, we define the optimal solution to LP~\ref{lp:private} that we need to prove the result:
	\begin{itemize}
		\item $\gamma_{i,\xi_i}=\sum_{s_i \in \sset_i} \sum_{\theta \in \Theta} \mu_\theta \phi_{i,\theta}(s_i) \gamma^{s_i}(\xi_i)$ for every $i \in \N$ and $\xi_i \in \Xi^q_i$.
		\item $t_{i,\xi_i,p_i}= \sum_{s_i \in \sset_i} \sum_{\theta \in \Theta} \mu_\theta \phi_{i,\theta}(s_i)\gamma^{s_i}(\xi_i)  \mathbb{I} \left\{ p_i=\hat p^s_i \right\} $ for every $i \in \N$, $\xi_i \in \Xi^q_i$, and $p_i \in P^b$.
		\item $y_{\theta,\xi,p}=\sum_{s \in \sset}  \mu_\theta \phi_\theta(s)  \prod_{i \in \N} \frac{\xi_i(\theta) \gamma^{s_i}(\xi_i)}{\xi_{i,s_i}(\theta)}  \mathbb{I}\{ p_i=\hat p^s_i \}$ for every $\theta \in \Theta$, $\xi \in \Xi^q$, and $p \in \mathcal{P}^b$.
	\end{itemize} 
	
	The next step is to show that, for every signal profile $s \in \sset$, the seller's expected revenue obtained by decomposing each signal $s_i$ according to Theorem~\ref{thm:decomposition} is ``close'' to the one for $s$.
	Formally, we show that, for every $s \in \sset$ and $\theta \in \Theta$, 
	\begin{align} \label{eq:revenueDecomposed}
		\sum_{\xi \in \Xi^q} \prod_{i \in \N} \frac{\xi_i(\theta) \gamma^{s_i}(\xi_i)}{\xi_{i,s_i}(\theta)}\textsc{Rev}(\V,\hat p^s,\xi) \ge \textsc{Rev}(\V,f(s),\xi_s)-\left( \alpha+\epsilon+\frac{1}{b} \right).
	\end{align}
	In order to do so, we relate the LHS of Equation~\eqref{eq:revenueDecomposed} to the seller' revenue in a non-Bayesian posted price auction.
	In particular, we show that it is equivalent to the seller's revenue when employing price vector $\hat p^s$ in the auction defined by the distributions of buyers' valuations $\hat \V^{s,\theta} = \{ \hat\V^{s,\theta}_i \}_{i \in \N}$, where each $\hat \V^{s,\theta}_i$ is such that a valuation $ v_i \sim \hat \V^{s,\theta}_i$ is defined as
	\[
		v_i = \tilde  v_i^\top \tilde \xi_i  \text{, with } \tilde v_i \sim \V_i \text{ and } \tilde \xi_i \in \Delta_{\Theta}\text{ sampled from a distribution such that }\pr \left\{ \tilde\xi_i=\xi_i \right\} = \frac{\xi_i(\theta)}{\xi_{i, s_i}(\theta)} \gamma^{s_i}(\xi_i).
	\]
	%
	%
	Notice that each $\hat \V_i^{s,\theta}$ is well defined, since $\V_i$ is by definition a probability distribution and $\sum_{\xi_i \in \Xi^q_i}\frac{\xi_i(\theta)}{\xi_{i,s_i}(\theta)}\gamma^{s_i}(\xi_i)=1$ by Theorem~\ref{thm:decomposition}, defining a probability distribution over the posteriors.
	Moreover, it is easy to check that valuations sampled from distributions $\hat \V_i^{s,\theta}$ are independent among each other.
	Finally, $\Rev(\hat \V^{s,\theta}, \hat p^s) $ is equal to to the LHS of Equation~\eqref{eq:revenueDecomposed}, since, by an inductive argument, for every $i \in \N$, it holds
	\begin{align*}
		\sum_{\xi' \in \Xi^q:\xi'_i=\xi_i} \prod_{j \in \N} \frac{\xi'_j(\theta) \gamma^{s_j}(\xi'_j)}{\xi_{j,s_j}(\theta)}  =\frac{\xi_i(\theta)}{\xi_{i,s_i}(\theta)}\gamma^{s_i}(\xi_i),
	\end{align*}
	where the equality comes from the fact that, for every $j \in \N$, it is the case that
	\begin{align}  \label{eq:sumToOne}
		\sum_{\xi_j \in \Xi^q_j} \frac{\xi_j(\theta) \gamma^{s_j}(\xi_j)}{\xi_{j,s_j}(\theta)} =1.
	\end{align}
	Let also notice that, in the auction defined above,
	the probability with which a buyer $i \in \N$ has a valuation greater than or equal to $p^s_i$ is $\sum_{\xi_i \in \Xi^q_i} \frac{\xi_i(\theta) \gamma^{s_i}(\xi_i)}{\xi_{i,s_i}(\theta)} \pr_{ \tilde v_i \sim \V_i}( \tilde v_i^{\top}  \xi_i\ge p^s_i)  \ge (1-\alpha) \pr_{\tilde  v_i \sim \V_i}( \tilde v_i^{\top} \xi_{i,s_i} \ge p^s_i)$, where the inequality holds by Equation~\eqref{eq:probBound}.
	First, we compare the seller's revenue in the two non-Bayesian, namely $\textsc{Rev}(\V^s,p^s)$ and $\textsc{Rev}( \hat \V^{s,\theta}, \hat p^s)$.
	In particular, we show by induction that
	$\textsc{Rev}( \hat \V^{s,\theta}, \hat p^s) \ge  \textsc{Rev}( \V^{s}, p^s)-\alpha- \epsilon-\frac{1}{b}$.
	Let $\textsc{Rev}_{\geq i}(\V,p)$ be the seller's expected revenue for $p$ in the auction restricted to all buyers $j \in \N: j \geq i$.
	The base case is
	\begin{align*}
		\textsc{Rev}_{\geq n}( \hat \V^{s,\theta}, \hat p^s)  & =  \hat p^s_n  \, \text{Pr}_{v_n \sim \hat \V^{s,\theta}_n} \left\{ v_n \geq \hat p^s_n   \right\}  \\
		& \geq \hat p^s_n (1-\alpha)\text{Pr}_{v_n \sim  \V^s_n} \left\{ v_n \geq p^s_n \right\} \\
		& \geq  p^s_n \, \text{Pr}_{v_n \sim  \V^s_n} \left\{ v_n \geq p^s_n \right\} - \epsilon -\alpha-\frac{1}{b}\\
		& =  \textsc{Rev}_{\geq n}( \V^s, p^s) - \epsilon-\alpha-\frac{1}{b}.
	\end{align*}
	
	By induction, let us assume that the condition holds for $i+1$, then
	\begin{align*}
		\textsc{Rev}_{\ge i}( \hat \V^{s,\theta}, \hat p^s) & = \hat p^s_i \text{Pr}_{v_i \sim \hat \V^{s,\theta}_i} \left\{ v_i \ge \hat p^s_i \right\} + \left( 1-\text{Pr}_{v_i \sim \hat \V^{s,\theta}_i} \left\{ v_i \ge \hat p^s_i \right\} \right) \textsc{Rev}_{>i}( \hat \V^{s,\theta}, \hat p^s) \\
		& \geq \left( p^s_i-\epsilon-\frac{1}{b} \right) \text{Pr}_{v_i \sim \hat \V^{s,\theta}_i} \left\{ v_i \ge \hat p^s_i \right\} + \left(1-\text{Pr}_{v_i \sim \hat \V^{s,\theta}_i} \left\{ v_i \ge \hat p^s_i \right\} \right) \, \left( \textsc{Rev}_{>i}(  \V^s,   p^s)-\epsilon -\alpha-\frac{1}{b} \right)  \\
		& =  p^s_i \text{Pr}_{v_i \sim \hat \V^{s,\theta}_i} \left\{ v_i \ge \hat p^s_i \right\} + \left(1-\text{Pr}_{v_i \sim \hat \V^{s,\theta}_i} \left\{ v_i \ge \hat p^s_i \right\} \right) \, (\textsc{Rev}_{>i}(  \V^s,  p^s)-\alpha)-\epsilon-\frac{1}{b} \\
		& \geq  p^s_i (1-\alpha) \text{Pr}_{v_i \sim   \V^s_i} \left\{ v_i \ge  p^s_i \right\} + \left[ 1-(1-\alpha)\text{Pr}_{v_i \sim   \V^s_i} \left\{ v_i \ge p^s_i \right\} \right] \, \left(\textsc{Rev}_{>i}(\V^s, p^s)-\alpha\right) - \epsilon-\frac{1}{b} \\
		&  \ge \textsc{Rev}_{\ge i}(  \V^s, p^s)- \epsilon-\alpha-\frac{1}{b},
	\end{align*}
	where the second to last inequality follows from $ p^s_i \geq \textsc{Rev}_{>i}( \V^s, p^s)$ and $\pr_{v_i \sim \hat \V^{s,\theta}_i} \left\{ v_i \ge \hat p^s_i \right\} \geq (1-\alpha)\pr_{v_i \sim \V^s_i} \left\{ v_i \ge  p^s_i \right\}$.
	Hence, Equation~\eqref{eq:revenueDecomposed} is readily proved, as follows
	\begin{align*} 
		\sum_{\xi \in \Xi^q} \prod_{i \in \N} \frac{\xi_i(\theta) \gamma^{s_i}(\xi_i)}{\xi_{i,s_i}(\theta)}\textsc{Rev}(\V,\hat p^s,\xi) \ge \textsc{Rev}( \V^s,p^s)-\left( \alpha+\epsilon+\frac{1}{b} \right)\ge \textsc{Rev}( \V,f(s),\xi_s)-\left(\alpha+\epsilon+\frac{1}{b} \right), 
	\end{align*}

	Now, we are ready to bound the objective of LP~\ref{lp:private}, as follows:
	\begin{align*}
		\sum_{\theta\in \Theta} \sum_{\xi\in \Xi^q} \sum_{p \in \mathcal{P}^b} y_{\theta,\xi,p} \textsc{Rev}(\V,p,\xi) & = \sum_{\theta\in \Theta} \sum_{\xi\in \Xi^q} \sum_{p \in \mathcal{P}^b} \sum_{s \in \sset}  \mu_\theta \phi_\theta(s)  \prod_{i \in \N} \frac{\xi_i(\theta) \gamma^{s_i}(\xi_i)}{\xi_{s_i}(\theta)}  \mathbb{I}\{ \hat p^s_i=p_i\}\textsc{Rev}(\V,p,\xi) \\ 
		&  = \sum_{s \in \sset} \sum_{\theta \in \Theta}   \mu_\theta \phi_\theta(s) \sum_{\xi\in \Xi^q}\sum_{p \in \mathcal{P}^bs}  \prod_{i \in \N} \frac{\xi_i(\theta) \gamma^{s_i}(\xi_i)}{\xi^{s_i}(\theta)}  \mathbb{I}\{ \hat p^s_i=p_i \} \textsc{Rev}(\V,p,\xi)  \\ 
		& \ge  \sum_{s \in \sset} \sum_{\theta \in \Theta} \mu_\theta \phi_\theta(s) \left[ \textsc{Rev}(\V,f(s),\xi_s)-\left(\alpha+\epsilon+\frac{1}{b}\right)\right]\\
		& \ge OPT-\left(\alpha+\epsilon+\frac{1}{b}\right)\ge OPT-\eta.
	\end{align*}
	
	
	We conclude the proof showing that the defined solution is feasible for LP~\ref{lp:private}.
	First, it prove that, for every $i \in \N$ and $\theta \in \Theta$,
	\begin{align*}
		\sum_{\xi_i \in \Xi^q_i} \xi_i(\theta) \gamma_{i,\xi_i}&
		=\sum_{\xi_i \in \Xi_i^q} \xi_i(\theta) \sum_{s_i \in \sset_i} \sum_{\theta'\in \Theta} \mu_{\theta'} \phi_{i,\theta'}(s_i) \gamma^{s_i}(\xi_i)\\
		& = \sum_{s_i \in \sset_i} \sum_{\theta' \in \Theta} \mu_{\theta'} \phi_{i,\theta'}(s_i) \sum_{\xi_i \in \Xi^q_i} \xi_i(\theta) \gamma^{s_i}(\xi_i)\\
		&= \sum_{s_i \in \sset_i} \sum_{\theta'\in \Theta} \mu_{\theta'} \phi_{i,\theta'}(s_i) \xi_{i,s_i}(\theta)\\
		&= \sum_{s_i \in \sset_i} \sum_{\theta' \in \Theta} \mu_{\theta'} \phi_{i,\theta'}(s_i) \frac{\mu_\theta \phi_{i,\theta}(s_i)}{\sum_{\theta' \in \Theta} \mu_{\theta'} \phi_{i,\theta'}(s_i)}\\
		& = \sum_{s_i \in \sset_i} \mu_\theta \phi_\theta(s_i)=\mu_\theta.
	\end{align*}
	Moreover, for every $i \in \N$ and $\xi_i \in \Xi^q_i$, it holds
	\begin{align*}
		\sum_{p_i \in P^b} t_{i,\xi_i,p_i}=&
		\sum_{p_i \in P^b} \sum_{s_i \in \sset_i} \sum_{\theta \in \Theta} \mu_\theta \phi_{i,\theta}(s_i)\gamma^{s_i}(\xi_i) \mathbb{I}\{ p_i=\hat p^s_i \}=\\
		&\sum_{s_i \in \sset_i} \sum_{\theta \in \Theta} \mu_\theta \phi_{i,\theta}(s_i)\gamma^{s_i}(\xi_i) \sum_{p_i \in P^b} \mathbb{I}\{ p_i=\hat p^s_i \}=\\
		&\sum_{s_i \in \sset_i} \sum_{\theta \in \Theta} \mu_\theta \phi_{i,\theta}(s_i)\gamma^{s_i}(\xi_i)=\gamma_{i,\xi_i}.
	\end{align*}
	Finally, for every $\theta \in \Theta$, $i \in \N$, $\xi_i \in \Xi^q_i$, and $p_i \in P^b$, it holds
		\begin{align*}
		\sum_{\xi' \in \Xi^q:\xi'_i=\xi_i} \sum_{p' \in \mathcal{P}^b:p'_i=p_i} y_{\theta,\xi',p'} & = \sum_{\xi' \in \Xi^q:\xi'_i=\xi_i} \sum_{p' \in \mathcal{P}^b:p'_i=p_i} \sum_{s \in \sset} \mu_\theta \phi_\theta(s) \prod_{j \in \N} \frac{\xi'_j(\theta)}{\xi_{j,s_j}(\theta)} \gamma^{s_j}(\xi'_j) \mathbb{I}\{ p_j=\hat p^s_j \}\\
		& = \sum_{s \in \sset} \mu_\theta \phi_\theta(s) \sum_{\xi' \in \Xi^q:\xi'_i=\xi_i} \sum_{p' \in \mathcal{P}^b:p'_i=p_i}  \prod_{j \in \N} \frac{\xi'_j(\theta)}{\xi_{j,s_j}(\theta)} \gamma^{s_j}(\xi'_j) \mathbb{I}\{p_j=\hat p^s_j\} \\
		& = \sum_{s \in \sset} \mu_\theta \phi_\theta(s) \frac{\xi_i(\theta)}{\xi_{s_i}(\theta)} \gamma^{s_i}(\xi_i)\mathbb{I}\{ p_i=\hat p^s_i \} \quad \quad\quad\quad \text{(From Equation~\eqref{eq:sumToOne})}\\
		& = \sum_{s_i \in \sset_i} \mu_\theta \phi_{i,\theta}(s_i) \frac{\xi_i(\theta)}{\xi_{s_i}(\theta)} \gamma^{s_i}(\xi_i) \mathbb{I}\{ p_i=\hat p^s_i \} \\
		& = \xi_i(\theta) \sum_{s_i \in \sset} \mu_\theta \phi_{i,\theta}(s_i) \frac{\sum_{\theta'\in \Theta} \phi_{i, \theta'}(s_i)}{\mu_{\theta} \phi_{i,\theta}(s_i)} \gamma^{s_i}(\xi_i)  \mathbb{I}\{ p_i=\hat p^s_i \}\\
		& = \xi_i(\theta) \sum_{s_i \in \sset} \sum_{\theta' \in \Theta} \mu_{\theta'} \phi_{i,\theta'}(s_i) \gamma^{s_i}(\xi_i) \mathbb{I}\{ p_i=\hat p^s_i \}= \xi_i(\theta) t_{i,\xi_i,p_i}.
	\end{align*}
	This concludes the proof.
	%
	%
\end{proof}

\begin{algorithm}[!htp]
	\caption{Approximate Dynamic Programming algorithm for MAX-LINREV}\label{alg:bbo3}
	\textbf{Inputs:} Discretization error tolerance $\delta > 0$; vector of linear components $w \in [0,1]^{n \times |\Xi^q_i| \times |P^b|}$; finite-support distributions of buyers' valuations $\V = \{ \V_i \}_{i \in \N}$
	\begin{algorithmic}[1]
		\State $c \gets \lceil \frac{n}{\delta} \rceil$
		\State $A \gets \{0,\frac{1}{c},\frac{2}{c},\ldots,\frac{n\delta}{c}\}$
		\State Initialize an empty matrix $M$ with dimension $n \times |A|$ \label{line:init1}
		\For{$a \in A$}
		\State $M(n, a) \gets \displaystyle\max_{\xi_n \in \Xi^q_n, p_n\in P^b :w_{n,\xi_n,p_n}\ge a}  \left\{ \Pr_{v_n \sim  \V_n} \left\{ v_n^\top \xi_n \ge p_n \right\} p_n \right\}$\label{line:inizialM}
		\EndFor
		\For{$i = n-1, \ldots, 1$ (in reversed order)}\label{line:forj3}
		\For{ $ a \in A$ }\label{line:fory3}
		\State $M(i,a) \gets \displaystyle\max_{\substack{\xi_i \in \Xi^q_i, p_i \in P^b, a' \in A : \\ w_{i,\xi_i,p_i}+a' \ge a}} \left\{ \Pr_{v_i \sim  \V_i} \left\{ v_i^\top \xi_i \ge p_i \right\} p_i + \left( 1-\Pr_{v_i \sim  \V_i} \left\{ v_i^\top \xi_i \ge p_i \right\} \right) M(i+1,a') \right\}$ \label{line:maxim}
		\EndFor
		\EndFor
		\State\Return $\max_{a \in A} \left\{ M(1,a)+a \right\}$
	\end{algorithmic}
\end{algorithm}

\lemmaeleven*

\begin{proof}
	The algorithm is described in Algorithm~\ref{alg:bbo3}.
	It works in polynomial time since the matrix $M$ has $n |A|= O(\frac{1}{\delta}n^3)$ entries and each entry is computed in polynomial time. This proves the second part of the statement.
	
	In the following, we denote with $\textsc{Rev}_{\geq i}(\V,p,\xi)$ the seller's expected revenue in the Bayesian posted price auction when they select price vector $p \in \mathcal{P}^b$ and the buyers' posteriors are specified by the tuple $\xi = (\xi_1, \ldots, \xi_n) \in \Xi^q$.
	
	Let $S(i,a) \coloneqq\{(\xi,p)\in\Xi^q\times \mathcal{P}^b \mid \sum_{j\ge i} w_{j,\xi_j,p_j}\ge a\}$ for every $i \in \N$ and $a \in A$.
	Moreover, for every $a' \in A$, let $\bar S(i,a,a')=\{(\xi,p)\in\Xi^q\times \mathcal{P}^b \mid w_{i,\xi_i,p_i}\ge a' \land \sum_{j> i} w_{j,\xi_j,p_j}\ge a-a'\}$.
	First, we prove by induction that $M(i,a-\frac{n-i}{c}) \ge \max_{(\xi,p) \in S(i,a)} \textsc{Rev}(\V,p,\xi)$ for every $i \in \N$ and $a \in A$.
	For $i=n$, the condition trivially holds by Line~\ref{line:inizialM}.
	For $i<n$, 
	\begin{align*}
		\max_{(\xi,p)\in S(i,a)} \textsc{Rev}_{\geq i}(\V,p,\xi) & =\max_{a' \in [0,1]} \max_{(\xi,p)\in \bar S(i,a,a')} \textsc{Rev}_{\geq i}(\V,p,\xi) \\
		&\hspace{-3cm}= \max_{a' \in [0,1]} \max_{(\xi,p)\in  \bar S(i,a,a')} p_i \Pr_{ v_i \sim  \V_i} \{  v_i^\top \xi_i\ge p_i \} + \left( 1-\Pr_{ v_i \sim  \V_i} \{  v_i^\top \xi_i\ge p_i \} \right) \textsc{Rev}_{\geq i+1}(\V,p,\xi)\\
		&\hspace{-3cm}\le \max_{a' \in A} \max_{(\xi,p) \in \bar S(i,a-\frac{1}{c},a')} p_i \Pr_{ v_i \sim  \V_i} \{  v_i^\top \xi_i\ge p_i \} + \left( 1-\Pr_{ v_i \sim  \V_i} \{  v_i^\top \xi_i\ge p_i \} \right) \textsc{Rev}_{\geq i+1}(\V,p,\xi)\\
		&\hspace{-3cm}= \max_{a' \in A} \,\, \max_{\xi_i \in \Xi^q_i,p_i \in P^b: w_{i,\xi_i,p_i}\ge a'} p_i \Pr_{ v_i \sim  \V_i} \{ v_i^\top \xi_i \ge p_i \} + \left( 1-\Pr_{ v_i \sim  \V_i} \{   v_i^\top \xi_i \ge p_i \} \right) \max_{(\xi,p) \in S(i+1,a-a'-\frac{1}{c})}\textsc{Rev}_{\geq i+1}(\V,p,\xi)\\
		&\hspace{-3cm} \le\max_{a' \in A} \max_{\xi_i \in \Xi^q_i,p_i \in P^b: w_{i,\xi_i,p_i}\ge a'} p_i \Pr_{ v_i \sim  \V_i} \{  v_i^\top \xi_i \ge p_i \} + \left( 1-\Pr_{ v_i \sim  \V_i} \{ v_i^\top \xi_i\ge p_i \} \right) M\left( i+1,a-a'-\frac{1}{c}-\frac{n-i-1}{c} \right)\\
		&\hspace{-3cm} = \max_{a' \in A,\xi\in \Xi^q,p \in \mathcal{P}^b:w_{i,\xi_i,p_i}+a'\ge a -\frac{n-i}{c}} p_i \Pr_{ v_i \sim  \V_i} \{  v_i^\top \xi_i\ge p_i \} + \left( 1-\Pr_{ v_i \sim  \V_i} \{ v_i^\top \xi_i\ge p_i \} \right) M(i+1,a') \\
		& \hspace{-3cm} = M \left( i,a-\frac{n-i}{c} \right).
	\end{align*}
	In conclusion, let $OPT_{\textsc{Rev}}$ the revenue term in the value of an optimal solution to MAX-LINREV, while $a\in [0,n]$ is the sum of the linear components in such optimal solution (the second term in the value of the solution).
	Let $a^*$ be the greatest element in $A$ such that $a^*\le a-\frac{n-1}{c}$. 
	Notice that $a^*\ge a-\frac{n}{c}$.
	Moreover, we have that $M(1,a^*)\ge \max_{(\xi,p) \in S(i,a)} \textsc{Rev}(\V,p,\xi)= OPT_{\textsc{Rev}}$.\footnote{It is easy to see that, if $a<\frac{n-1}{c}$, then the equality holds for $a=0$.}
	Hence, there exists a solution with value $M(1,a^*)+a^* \ge OPT_{\textsc{Rev}} +a- \frac{n}{c} =OPT -\frac{n}{c}\ge OPT-\delta$, concluding the proof.
\end{proof}

\theoremfour*

\begin{proof}
	We start providing the following relaxation of LP~\ref{lp:private}:
	\begin{subequations} \label{lp:privategreater}
		\begin{align}
		\max_{\substack{\gamma, x, y\ge 0}} & \,\,
		\sum_{\theta \in \Theta} \sum_{\xi \in \Xi^q} \sum_{p \in \mathcal{P}^b} y_{\theta,\xi,p} \textsc{Rev}(\V,\xi,p) \quad \textnormal{s.t.} \\
		& \xi_i(\theta) t_{i,\xi_i,p_i} \ge \sum_{\xi' \in \Xi : \xi'_i = \xi_i} \sum_{p' \in \mathcal{P}^b : p'_i = p_i} y_{\theta, \xi',p'} &  \forall \theta \in \Theta, \forall i \in \N, \forall \xi_i \in \Xi_i^q, \forall p_i \in P^b \\
		&\sum_{p_i \in P^b} t_{i,\xi_i,p_i} = \gamma_{i,\xi_i} & \forall i \in \N, \forall \xi_i \in \Xi_i^q   \\
		& \sum_{\xi_i \in \Xi^q} \gamma_{i , \xi_i} \, \xi_i(\theta) =\mu_\theta & \forall i \in \N, \forall \theta \in \Theta. 
		\end{align}
	\end{subequations}
	The PTAS that we build in the rest of the proof works with the dual of LP~\ref{lp:privategreater} so as to take advantage of the fact that it is more constrained than that of the original LP~\ref{lp:private}.
	As a first step, the following lemma shows that LP~\ref{lp:private} and LP~\ref{lp:privategreater} are equivalent.
	\begin{restatable}{lemma}{lemmaten} \label{lm:relaxEqual}
		LP~\ref{lp:private} and LP~\ref{lp:privategreater} have the same optimal value.
		Moreover, given a feasible solution to LP~\ref{lp:privategreater}, it is possible compute in polynomial time a feasible solution to LP~\ref{lp:private} with a greater or equal value.
	\end{restatable}
	\begin{proof}
		To show the equivalence between the two LPs, it is sufficient to show that, given a feasible solution to LP~\ref{lp:privategreater}, we can construct a solution to LP \ref{lp:private} with a greater or equal value.
		Let $(y,t,\gamma)$ be a solution to LP~\ref{lp:privategreater}.
		For every $i \in \N$, $\xi_i \in \Xi^q_i$, $p_i \in P^b$, let $\delta_{i,\xi_i,p_i} \coloneqq \xi_i(\theta) t_{i,\xi_i,p_i}-\sum_{\xi' \in \Xi^q:\xi'_i=\xi_i} \sum_{p' \in \mathcal{P}^b:p'_i=p_i}  y_{\theta,\xi',p'}$.
		Moreover, let $\iota=\mu_\theta-  \sum_{\xi \in \Xi^q, p \in \mathcal{P}^b}  y_{\theta,\xi,p}$.
		First, we show that $\sum_{\xi_i \in \Xi^q_i,p_i \in P^b}  \delta_{i,\xi_i,p_i} =\iota$ for every $i \in \N$.
		For each $i \in \N$, it holds
		\begin{align*}
		\sum_{\xi_i \in \Xi^q_i} \sum_{p_i \in P^b}  \delta_{i,\xi_i,p_i}& = \sum_{\xi_i \in \Xi^q_i}\sum_{p_i \in P^b} \left[\xi_i(\theta) t_{i,\xi_i,p_i}- \sum_{\xi' \in \Xi^q:\xi'_i=\xi_i} \sum_{p' \in \mathcal{P}^b:p'_i=p_i} y_{\theta,\xi',p'}\right] \\
		& = \sum_{\xi_i \in \Xi^q_i } \xi_i(\theta) \gamma_{i,\xi_i}-   \sum_{\xi' \in \Xi^q} \sum_{p' \in \mathcal{P}^b}  y_{\theta,\xi',p'} \\
		&  = \mu_\theta- \sum_{\xi \in \Xi^q} \sum_{p \in \mathcal{P}^b}  y_{\theta,\xi,p} = \iota.
		\end{align*}
		Next, we build a feasible solution $(\bar y,t,\gamma)$ to LP~\ref{lp:private} with $\bar y_{\theta,\xi,p}\ge y_{\theta,\xi,p}$ for all $\theta \in \Theta$, $\xi \in \Xi^q$, and $p \in \mathcal{P}^b$.
		In particular, we set $\bar y_{\theta,\xi,p} = y_{\theta,\xi,p}+ \frac{\prod_{i \in \N}  \delta_{i,\xi_i,p_i}}{\iota^{n-1}}$.
		Since $\delta_{i,\xi_i,p_i}\ge 0$ and $\iota \ge 0$ by the feasibility of $(y,t,\gamma)$, it holds that $\bar y_{\theta,\xi,p}\ge y_{\theta,\xi,p} $.
		Moreover, for each $i \in \N$, $\theta \in \Theta$, $\xi_i \in \Xi_i^q$, and $p_i \in P^b$, we have that
		\begin{align*}
				\sum_{\xi' \in \Xi^q:\xi'_i=\xi_i} \sum_{p' \in \mathcal{P}^b:p'_i=p_i} \bar y_{\theta,\xi',p'} & =\sum_{\xi' \in \Xi^q:\xi'_i=\xi_i} \sum_{p' \in \mathcal{P}^b:p'_i=p_i}  y_{\theta,\xi',p'}+\sum_{\xi' \in \Xi^q:\xi'_i=\xi_i} \sum_{p' \in \mathcal{P}^b:p'_i=p_i} \frac{\prod_{j \in \N}  \delta_{j,\xi_j,p_j}}{\iota^{n-1}} \\
				& = \sum_{\xi' \in \Xi^q:\xi'_i=\xi_i} \sum_{p' \in \mathcal{P}^b:p'_i=p_i}  y_{\theta,\xi,p'}+ \delta_{i,\xi_i,p_i}=t_{i,\xi_i,p_i},
		\end{align*}
		%
		where the second equality follows from $ \sum_{\xi_j \in \Xi^q_j,p_j \in P^b} \delta_{j,\xi_j,p_j} =\iota$ for every $j \in \N$.
		Since $\textsc{Rev}(\V,p,\xi)\ge 0$ for every $p \in \mathcal{P}^b$ and $\xi \in \Xi^q$, it follows that the value of $(\bar y,t,\gamma)$ is greater than or equal to the value of $(\bar y,t,\gamma)$.
	\end{proof}

	Our PTAS is described in Algorithm~\ref{alg:bisection}.
	
	\begin{algorithm}[H]\caption{PTAS for the private signaling setting}
		\textbf{Input:} Error $\beta$, approximation factor of the approximation oracle $\delta$, $q$ defining the set of q-uniform posteriors, \# of discretization steps $b$, number of samples $K$.
		\begin{algorithmic}[1]
			\State \textbf{Initialization}: $\rho_1\gets0$, $\rho_2\gets1$,  $H \gets\varnothing$,$H^* \gets\varnothing$.
			\State obtain an empirical distribution of valuations $\V^K$ sampling $K$ samples from $\V$.
			\While{ $\rho_2-\rho_1>\beta$}
			\State $\rho_3\gets (\rho_1+\rho_2)/2$
			\State $H \gets\{\text{\normalfont violated constraints returned by the ellipsoid method on } \circled{F} \,\, \text{\normalfont with objective } \rho_3 \,\text{\normalfont and approximation error}\, \delta \}$
			\If  {unfeasible}
			\State $\rho_1\gets \rho_3$
			\State $H^*\gets H$
			\Else{
				\State $\rho_2 \gets \rho_3$}
			\EndIf
			\State \textbf{return} the solution to LP \ref{lp:reducedPrimal} with only constraints in $H^*$
			\EndWhile
		\end{algorithmic}	\label{alg:bisection}
	\end{algorithm}
	
	Since we only have access to an oracle returning samples from distributions $\V$, our algorithm works with empirical distributions $\V^K$ built from $K$ i.i.d samples, for a suitably-defined $K \in \mathbb{N}_{>0}$.
	%
	The algorithm works with LP~\ref{lp:privategreater} for the values $b \in \mathbb{N}_{>0}$ and $q \in \mathbb{N}_{>0}$ defined in the following, finding an approximate solution to LP~\ref{lp:privategreater}.
	Since LP~\ref{lp:privategreater} has an exponential number of variables, the algorithm works by applying the ellipsoid method to its dual formulation, as described in the following.
	
	Let $a \in \mathbb{R}^{n \times |\Theta|}$, $w \in \mathbb{R_{-}}^{ |\Theta| \times  n \times |\Xi_i^q| \times |P^b|}$, and  $c \in \mathbb{R}^{  n \times |\Xi^q_i|}$. Then, the dual of LP \ref{lp:privategreater} reads as follow.
	\begin{subequations}\label{lp:privatedual}
		\begin{align}
			\min_{\substack{a,w\le0,c}} & \,\, \sum_{i \in \N}\sum_{\theta \in \Theta} \mu_\theta a_{i,\theta} \quad \textnormal{s.t.} \\
			&\sum_{i \in \N}   -w_{\theta,i,\xi_i,p_i}\ge \textsc{Rev}(\V^K,p,\xi) & \forall \theta \in \Theta,\forall \xi \in \Xi^q, \forall p \in \mathcal{P}^b \label{eq:dual1}\\
			& \sum_{\theta \in \Theta}  \xi_i(\theta) w_{\theta,i,\xi_i,p_i} + c_{i,\xi_i} \ge 0 & \forall i \in \N,\forall\xi_i \in \Xi^q_i,\forall p_i \in P^b \label{eq:dual2}\\
			& - c_{i,\xi_i}+ \sum_{\theta \in \Theta} \xi_i(\theta) a_{i,\theta} \ge 0 & \forall i \in \N,\forall \xi_i \in \Xi^q_i.  \label{eq:dual3}
		\end{align}
	\end{subequations}
	Notice that, by using the dual of LP~\ref{lp:privategreater} instead of that of LP~\ref{lp:private}, we get the additional constraint $w\le0$.
	LP~\ref{lp:privatedual} has a polynomial number of variables and a polynomial number of Constraints~\eqref{eq:dual2}~and~\eqref{eq:dual3}.
	Hence, to solve the LP using the ellipsoid method we need a separation oracle for Constraints~\eqref{eq:dual1}, which are exponentially many.
	Instead of an exact separation oracle, we use an approximate separation oracle that employs Algorithm~\ref{alg:bbo3} with a suitably-defined $\delta>0$.
	We use a binary search scheme to find a value $\rho^\star\in [0,1]$ such that the dual problem with objective $\rho^\star$ is unfeasible, while the dual with objective $\rho^\star+\beta$ is \emph{approximately} feasible, for some $\beta \geq 0$ defined in the following.
	The algorithm requires $\log(\beta)$ steps and, at each step, it works by determining, for a given value $\rho_3$, whether there exists a feasible solution for the following feasibility problem that we call \circled{F}:
	\begin{subequations}
	\begin{align} \label{lp:privatedual1}
	& \sum_{i \in \N} \sum_{\theta \in \Theta}\mu_\theta a_{i,\theta} \le \rho_3\\
	& \textsc{Rev}(\V,p,\xi)+\sum_{i \in \N}   w_{\theta,i,\xi_i,p_i}\le 0 &  \forall \theta \in \Theta, \forall\xi \in \Xi^q, \forall p \in \mathcal{P}^b \\
	& \sum_{\theta \in \Theta}  \xi_i(\theta) w_{\theta,i,\xi_i,p_i} + c_{i,\xi} \ge 0 & \forall i \in \N, \forall \xi_i \in \Xi^q_i, \forall p_i \in P^b \\
	&- c_{i,\xi_i}+ \sum_{\theta \in \Theta} \xi_i(\theta) a_{i,\theta} \ge 0 & \forall i \in \N, \forall \xi_i \in \Xi^q_i \\
	& w_{\theta,i,\xi_i,p_i}\le 0 & \forall \theta \in \Theta , \forall i \in \N, \forall \xi_i \in \Xi^q_i, \forall  p_i \in P^b.
	\end{align}
	\end{subequations}
	At each iteration of the bisection algorithm, the feasibility problem \circled{F} is solved via the ellipsoid method.
	To do so, we need a separation oracle.
	We use an approximate separation oracle that returns a violated constraint that will be defined in the following.
	The bisection procedure terminates when it determines a value $\rho^\star$ such that on \circled{F} the ellipsoid method returns unfeasible for $\rho^\star$, while returning feasible for $\rho^\star+\beta$.
	Finally, the algorithm solves a modified primal LP \ref{lp:reducedPrimal} with only the subset of variables $y$ in $H^*$, where $H^*$ is the set of violated constraints returned by the ellipsoid method applied on the unfeasible problem with objective $\rho^*$.
	From this solution, we can use Lemma \ref{lm:relaxEqual} to find a solution to LP~\ref{lp:private} with the same value and Lemma~\ref{lm:lpToSignaling} to find a signaling scheme with the same seller's revenue as the value of the solution.
	
	\paragraph{Approximate Separation Oracle.}
	Our separation oracle works as follow. Given a point $(a,w,c)$ in the dual space, we check if a constraint relative to the variables $t$ and $\gamma$ of the primal is violated. Since there are a polynomial number of these constraints, it can be done in polynomial time. If it is the case, we return that constraint.
	Otherwise, our idea is to use Algorithm \ref{alg:bbo3} with a $\delta$ defined in the following to find if a constraint relative to variable $y$ is violated.
	We apply Algorithm \ref{alg:bbo3}, once for each possible state $\theta \in \Theta$.
	In the following, we assume that $\theta$ is fixed and we denote $w_{\theta,i,\xi_i,p_i}$ as $w_{i,\xi_i,p_i}$.
	Algorithm~\ref{alg:bbo3} needs values such that $w_{i,\xi_i,p_i} \in[0,1]$ for all $i \in \N$, $\xi_i \in \Xi^q_i$, and $p_i \in P^b$.
	We show that we can restrict the inputs to $w_{i,\xi_i,p_i} \in[-1,0]$.\footnote{It is easy to see that summing 1 to all the elements of the vector $w$ does not change the problem.}
	By constraint $w\le0$, all $w_{i,\xi_i,p_i} $ are non-positive. Otherwise, this constraint is violated and would have been returned in the first step.
	Moreover, given a vector $w$, we give as input to the oracle a vector $\bar w$ such that $\bar w_{i,\xi_i,p_i}=-1$ whenever $w_{i,\xi_i,p_i}<-1$.

	If for at least one state $\theta$ a violated constraint is found by Algorithm~\ref{alg:bbo3}, we return that constraint, otherwise we return feasible.
	Our separation oracle has two properties. When it returns a violated constraint, the constraint is actually violated. In particular, if $\sum_{i \in \N} \bar w_{\theta,i,\xi_i,p_i}+ \Rev(\V^K,p,\xi)>0$, then $\bar w_{\theta,i,\xi_i,p_i}>-1$ for every $i \in \N$, implying $\bar w_{\theta,i,\xi_i,p_i}= w_{\theta,i,\xi_i,p_i} $ and  $\sum_{i \in \N}  w_{\theta,i,\xi_i,p_i}+ \Rev(\V^K,p\xi)= \sum_{i \in \N}  \bar w_{\theta,i,\xi_i,p_i}+ \Rev(\V^K,p,\xi)>0$,
	Additionally, when the separation oracle returns feasible, then all the constraints relative to the variables $y$ are violated by at most $\delta$.
	Suppose by contradiction that a constraint for a triple $(\theta, \xi,p)$ is violated by more than $\delta$. Then, the separation oracle would have found $\theta^* \in \Theta$, $\xi^*\in \Xi^q$, and $p^* \in \mathcal{P}^b$ such that:
	$\sum_{i \in \N} \bar w_{\theta^*,i,\xi^*_i,p^*_i}+ \Rev(\V^K,p^*,\xi^*)\ge  \sum_{i \in \N} \bar w_{\theta,i,\xi_i,p_i}+ \Rev(\V^K,p,\xi)-\delta\ge \sum_{i \in \N}  w_{\theta,i,\xi_i,p_i}+ \Rev(\V^K,p,\xi)-\delta>0$, and, thus, it would have returned this violated constraint.

	\paragraph{Approximation Guarantee.}
	The algorithm finds a $\rho^*$ such that the problem is unfeasible, \emph{i.e.}, the value of $\rho_1$ when the algorithm terminates, and a value smaller than or equal to $\rho^*+\beta$ such that the ellipsoid method returns feasible, \emph{i.e.}, the value of $\rho_2$ when the algorithm terminates.
	For each possible distribution of the samples $\V^K$, let $OPT^{\V^K}$ be the optimal value of LP~\ref{lp:privatedual}.
	As a first step, we bound the value of $OPT^{\V^K}$.
	In particular, we show that $OPT^{{\V^K}}\le \rho^*+\beta+\delta$.
	Since, the bisection algorithm returns that \circled{F} is feasible with objective $\rho^*+\beta$, it finds a solution $(a,w,c)$ such that all the constraints regarding variables $t$ and $\gamma$ of the primal are satisfied and the approximate separation oracle did not find a violated constraint for the constraints regarding variables $y$.
	We show that $(a,w,c)$ is a solution to the following LP.
	\begin{subequations}\label{lp:privatedual2}
		\begin{align}
			& \sum_{i \in \N}\sum_{\theta \in \Theta} \mu_\theta a_{i,\theta} \le \rho^*+\beta \\
			& \sum_{i \in \N}   -w_{\theta,i,\xi_i,p_i}\ge \textsc{Rev}(\V^k,p,\xi) -\delta &  \forall \theta \in \Theta, \forall\xi \in \Xi^q, \forall p \in \mathcal{P}^b \label{eq:modifieddual1}\\
			& \sum_{\theta\in \Theta}  \xi_i(\theta) w_{\theta,i,\xi_i,p_i} + c_{i,\xi_i} \ge 0 &  \forall i \in \N, \forall \xi_i \in \Xi^q_i, \forall p_i \in P^b \label{eq:modifieddual2} \\
			& - c_{i,\xi_i}+ \sum_{\theta \in \Theta} \xi_i(\theta) a_{i,\theta} \ge 0 & \forall i \in \N, \forall \xi_i \in \Xi^q_i\label{eq:modifieddual3} \\
			&w_{\theta,i,\xi_i,p_i}\le 0 & \forall \theta \in \Theta, \forall i \in \N, \forall \xi_i \in \Xi^q_i,\forall p_i \in P^b.
		\end{align}
	\end{subequations}
	All the Constraints \eqref{eq:modifieddual2} and \eqref{eq:modifieddual3} are satisfied since the separation oracle checks them explicitly, while we have shown that, when the separation oracle return feasible, it holds $\sum_{i \in \N} w_{\theta,i,\xi_i,p_i}+ \textsc{Rev}(\V^K,p,\xi) \le \delta $ for all $ \theta \in \Theta, \xi \in \Xi^q_i,$ and $p \in \mathcal{P}^b$, implying that all the Constraints \eqref{eq:modifieddual1} are satisfied.
	
	Then, by strong duality the value of the following LP is at most $\rho^*+\beta$.
	
	\begin{subequations} \label{lp:modifiedprimal}
		\begin{align}
			\max_{y,t,\gamma} & \,\, \sum_{\theta \in \Theta} \sum_{\xi \in \Xi^q} \sum_{p \in \mathcal{P}^b} y_{\theta,\xi,p} \left( \textsc{Rev}(\V^k,p,\xi)-\delta \right) \quad \textnormal{s.t.}\\
			& \xi_i(\theta) t_{i,\xi_i,p_i} \ge \sum_{\xi' \in \Xi^q:\xi'_i=\xi_i} \sum_{ p' \in \mathcal{P}^b:p'_i=p_i}  y_{\theta,\xi,p} & \forall	\theta \in \Theta, \forall i \in \N, \forall \xi_i \in \Xi^q_i, \forall p_i \in P^b  \\
			&\sum_{p \in P^b} t_{r,\xi_i,p} = \gamma_{i,\xi_i} & \forall i \in \N, \forall \xi_i \in \Xi^q_i  \\
			& \sum_{\xi_i \in \Xi^q_i} \gamma_{i,\xi_i} \xi_i(\theta) =\mu_\theta & \forall i \in \N, \forall \theta \in \Theta. \\ 
		\end{align}
	\end{subequations}
	
	Notice that any solution to LP \ref{lp:privategreater} is also a feasible solution to the previous modified problem.
	Since in any feasible solution $\sum_{\theta \in \Theta} \sum_{\xi \in \Xi^q, p \in \mathcal{P}^b} y_{\theta,\xi,p}=1$ and LP \ref{lp:modifiedprimal} has value at most $\rho^* +\beta$, then $OPT^{\V^K} \le \rho^* +\beta +\delta$.

	Let $H^*$ be the set of constraints regarding variables $y$ returned by the ellipsoid method run with objective $\rho^\star$.
	Since the ellipsoid method with the approximate separation oracle returns unfeasible, by strong duality LP \ref{lp:privategreater} with only the variables $y$ relative to constraints in $H^*$ has value at least $\rho^*$. Moreover, since the ellipsoid method guarantees that $H^*$ has polynomial size, the LP can be solved in polynomial time.
	Hence, solving the following LP, \emph{i.e.}, the primal LP \ref{lp:privategreater} with only the variables $y$ in $H^*$, we can find a solution with value at least $\rho^*$.
	
	\begin{subequations} \label{lp:reducedPrimal}
		\begin{align}
			\max_{\substack{\gamma, t, y\ge 0}} & \,\, 
			\sum_{\theta \in \Theta} \sum_{(\xi,p):(\theta,\xi,p)\in H^*} y_{\theta,\xi,p} \textsc{Rev}(\V^K,p,\xi) \quad \textnormal{s.t.} \\
			& \xi_i(\theta) t_{i,\xi_i,p_i} \ge\sum_{\xi',p': (\theta,\xi',p') \in H^*:  \xi'_i=\xi_i, p'_i=p_i,}  y_{\theta,\xi',p'} & \forall \theta \in \Theta, \forall i \in \N, \forall \xi_i \in \Xi_i^q, \forall p_i \in P^b \\
			&\sum_{p_i \in P^b} t_{i,\xi_i,p_i} = \gamma_{i,\xi_i} & \forall i \in \N, \forall \xi_i \in \Xi^q_i   \\
			& \sum_{\xi_i \in \Xi^q_i} \gamma_{i,\xi_i} \xi_i(\theta) =\mu_\theta &  \forall i \in \N,\forall \theta \in \Theta  .
		\end{align}
	\end{subequations}

	To conclude the proof, we show that replacing the distributions $\V$ with $ \V^K$, the expected revenue decreases by a small amount.
	Let $y^{APX,\V^K}$ be the solution returned by the algorithm with distribution $\V^K$.
	Moreover, let $y^{OPT,\V^K}$ be the optimal solution to LP \ref{lp:privategreater} with distributions $\V^K$ and $y^{OPT, \V}$ the optimal solution with distributions $\V$.
	Finally, let  $OPT$ be the value of the optimal private signaling scheme with distributions $\V$.
	
	Let $\epsilon$ be a constant defined in the following and  $K=8 \log(2|\Xi^q||\mathcal{P}^b|/\epsilon)/\epsilon^2$.
	By Hoeffding bound, for every $\xi \in \Xi^q$ and $p \in \mathcal{P}^b$, with probability at least $1-e^{-2K/(\epsilon/4)^2}=1-|\Xi^q||\mathcal{P}^b|\epsilon/4$, 
	\[|\textsc{Rev}(\V,p,\xi)|-|\textsc{Rev}(\V^K,p,\xi)|\le \epsilon/4.\] By the union bound, it implies that with probability at least $1-\epsilon/2$, $|\textsc{Rev}(\V,p,\xi)|-\textsc{Rev}(\V^K,p,\xi)|\le \epsilon/2$ for every $\xi \in \Xi^q$ and $p \in \mathcal{P}^b$.
	Then, with probability $1- \epsilon/2 $, 
	\begin{align*}
	 &\sum_{\theta \in \Theta} \sum_{\xi \in \Xi^q} \sum_{p \in \mathcal{P}^b}  y^{APX,\V^K}_{\theta,\xi,p} \textsc{Rev}(\V,p,\xi)\ge\\
	 & \sum_{\theta \in \Theta} \sum_{\xi \in \Xi^q} \sum_{p \in \mathcal{P}^b}   y^{APX,\V^K}_{\theta,p,\xi} \textsc{Rev}(\V^K,p,\xi) -\epsilon/4\ge \\
	 &\sum_{\theta \in \Theta} \sum_{\xi \in \Xi^q} \sum_{p \in \mathcal{P}^b}   y^{OPT, \V^K}_{\theta,p,\xi} \textsc{Rev}(\V^K,p,\xi) -\epsilon/4 -\delta-\beta \ge \\
	 & \sum_{\theta \in \Theta} \sum_{\xi \in \Xi^q} \sum_{p \in \mathcal{P}^b}   y^{OPT, \V}_{\theta,p,\xi} \textsc{Rev}(\V^K,p,\xi) -\epsilon/4 -\delta-\beta \ge \\
	& \sum_{\theta \in \Theta} \sum_{\xi \in \Xi^q} \sum_{p \in \mathcal{P}^b}   y^{OPT, \V}_{\theta,p,\xi} \textsc{Rev}\V,p,\xi) -\epsilon/2 -\delta-\beta  \ge \\
	& OPT-\epsilon/2 -\delta-\beta-\eta
	\end{align*}
		Hence, with probability  $1-\epsilon/2$, the solution has value at least $OPT-\epsilon/2-\delta-\beta -\eta$ and 
		\[\mathbb{E}_{\V^K} \left[\sum_{\theta \in \Theta} \sum_{\xi \in \Xi^q} \sum_{p \in \mathcal{P}^b}  y^{APX,\V^K}_{\theta,\xi,p} Rev(\V,p,\xi) \right] \ge OPT-\epsilon/2-\epsilon/2-\delta-\beta -\eta=OPT-\epsilon-\delta-\beta -\eta,\]
	 where the expectation is on the sampling procedure.
	
	To conclude the proof, to have an approximation error $\lambda$, we can set $b$ and $q$ such that the approximation error in Lemma \ref{lm:privateQuniform} is $\eta=\lambda/4$ and $\epsilon=\delta=\beta=\lambda/4$.
	Finally, given an approximate solution to LP \ref{lp:reducedPrimal}, Lemma \ref{lm:relaxEqual} provides a solution to LP \ref{lp:private} with greater or equal value and Lemma~\ref{lm:lpToSignaling} recover a signaling scheme with the same revenue.
\end{proof}

\end{document}